\documentclass[12pt,a4paper]{amsart}
\usepackage{amssymb}
\usepackage[a4paper,left=2.25cm,right=2.25cm,top=3cm,bottom=3cm,headsep=1cm]{geometry}
\usepackage{tikz}\usetikzlibrary{matrix,fit}
\usepackage{enumerate}
\usepackage{hyperref}
\usepackage{mathtools}

\newcommand\cf{\text{cf}}
\newcommand{\At}{\widetilde{A}}
\newcommand{\Bt}{\widetilde{B}}
\newcommand\Qt{\widetilde Q}
\newcommand\wt[1]{\text{weight}(#1)}
\newcommand{\bea}{\begin{equation}}
\newcommand{\eea}{\end{equation}}
\renewcommand{\ss}{\scriptstyle}

\let\Im\relax
\DeclareMathOperator\Im{Im}
\DeclareMathOperator\Ker{Ker}
\DeclareMathOperator\Hom{Hom}
\DeclareMathOperator\End{End}
\DeclareMathOperator\tr{tr}
\DeclareMathOperator\rank{rank}
\DeclareMathOperator\Spec{Spec}

\newtheorem{lem}{Lemma}
\newtheorem*{lem*}{Lemma}
\newtheorem*{rmk}{Remark}

\newtheorem{prop}{Proposition}
\newtheorem{conj}{Conjecture}
\newtheorem{cor}{Corollary}

\title[Higher spin $R$-matrix from equivariant (co)homology]{Higher spin $\mathfrak{sl}_2$ $R$-matrix from equivariant (co)homology}

\author{Dmitri Bykov}
\address{Dmitri Bykov\newline
Max-Planck-Intitut f\"ur Physik, F\"ohringer Ring 6, D-80805 Munich, Germany\newline
Arnold Sommerfeld Center for Theoretical Physics,  Theresienstrasse 37, D-80333 Munich, Germany\newline
Steklov Mathematical Institute of Russ.~Acad.~Sci., Gubkina str. 8, 119991 Moscow, Russia}
\email{bykov@mpp.mpg.de, bykov@mi-ras.ru}
\author{Paul Zinn-Justin}
\address{Paul Zinn-Justin\newline
School of Mathematics and Statistics, The University of Melbourne, 
Victoria 3010, Australia}
\email{pzinn@unimelb.edu.au}
\thanks{PZJ was supported by ARC grant FT150100232. DB wishes to thank D.~L\"ust and A.~A.~Slavnov for support. 
PZJ~wishes to thank A.~Knutson, Y.~Yang, G.~Zhao for valuable discussions.
Computerized checks of the results of this paper were performed
with the help of Macaulay2~\cite{M2}.
}

\newcommand\rem[2][]{}

\newcommand\M{\mathfrak M}
\renewcommand\P{\mathfrak P}
\renewcommand\S{\mathfrak S}
\newcommand\Z{\mathfrak Z}
\newcommand\CC{\mathbb C}
\newcommand\RR{\mathbb R}
\newcommand\PP{\mathbb P}
\newcommand\ZZ{\mathbb Z}

\begin{document}

\maketitle

\begin{abstract}
We compute the rational $\mathfrak{sl}_2$ $R$-matrix acting in the product of two spin-$\ell\over 2$ (${\ell \in \mathbb{N}}$) representations, using a method analogous to the one of Maulik and Okounkov, i.e., by studying the equivariant (co)homology of certain algebraic varieties. These varieties, first considered by Nekrasov and Shatashvili, are typically singular.  They may be thought of as the higher spin generalizations of $A_1$ Nakajima quiver varieties (i.e., cotangent bundles
of Grassmannians), the latter corresponding to $\ell=1$.
\end{abstract}


\section{Introduction}
In \cite{Nakaj-quiv3, Varagnolo} the cohomology of quiver varieties introduced in~\cite{Nakaj-quiv1, Nakaj-quiv2} was endowed with the action of a Yangian. More precisely, one obtains arbitrary tensor products of fundamental representations this way.
In \cite{MO-qg}, Maulik and Okounkov reinterpreted this construction in the formalism of quantum integrable systems:
via the {\em stable envelope}, they defined the {\em $R$-matrix}\/ which in turn allows to define the Yangian action
thanks to the $RTT$ approach, as first advocated in \cite{KR} (see also~\cite{Drinfeld}).

It is known that any finite-dimensional irreducible representation
of the Yangian of a simple Lie algebra
can be obtained as a subquotient of tensor products of fundamental representations \cite[Cor.~12.1.13]{CP-book}. One could argue that defining geometrically
such tensor products of fundamental representations is all that is therefore needed. We believe that this point of view is
unsatisfactory: one should be able to obtain directly such general finite-dimensional representations as the cohomology of
some appropriate varieties generalizing Nakajima quiver varieties. The relation to tensor products of fundamental representations
(usually called {\em fusion procedure}\/ in the language of quantum integrable systems) should come about as a {\em Lagrangian
correspondence}, which itself should be a degenerate limit of the stable envelope construction (see also \cite{artic64} where fusion between different Nakajima quiver varieties is considered).

This paper is concerned with the case of the simplest quiver -- $A_1$. It should be considered as a preliminary study, 
since we do not try to formulate a general formalism \`a la Maulik--Okounkov (in particular, it does not directly
fit in the framework of the latter since our varieties are singular). 
We go back to the source of (or at least one of the motivations for) the correspondence between quantum integrable systems
and cohomology theories, namely the work of Nekrasov and Shatashvili \cite{NS-short}. They consider a supersymmetric gauge theory in two-dimensional spacetime, whose space of classical groundstates is an algebraic variety, naturally defined in terms of a certain K\"ahler quotient. At the quantum level, and in the presence of the so-called twisted masses, the theory has a finite number of supersymmetric ground states. These ground states may be distinguished by the expectation values of certain scalar operators  (which are the superpartners of the gauge fields) that satisfy a set of algebraic equations, which miraculously may be identified with the Bethe equations. It was conjectured in~\cite{NS-short} that geometrically these Bethe equations may be interpreted as the relations in the quantum equivariant cohomology ring of the original algebraic variety. In this setup the twisted masses play the role of equivariant parameters. The simplest example of this correspondence (the $A_1$-quiver) is given by the cotangent bundle
of Grassmannians, in which case the Bethe equations are those of the homogeneous Heisenberg spin-$1\over 2$ chain. In~\cite{NS-short} one also finds a proposal for the case of arbitrary spin. 
As we shall see, this case is substantially more complicated, and the algebraic varieties that arise are singular.
We study in this setting the concept of stable envelope, 
which is central to \cite{MO-qg}, 
and compute the $R$-matrix explicitly,
based on the stable envelope construction in size $n=2$. Note that the whole integrable structure follows from the $R$-matrix:
in particular, the Yangian action itself can be defined from it in the $RTT$ approach, 
cf \cite[Sect.~5]{MO-qg}, and we shall not repeat here those general arguments.

For simplicity, we only consider here the case of ``pure'' spins (each site of the chain has the same spin). Presumably,
the same approach works for mixed spins.

The structure of the paper is as follows. In section~\ref{sec:NS} we recall the setup of Nekrasov and Shatashvili~\cite{NS-short}, who derived the Bethe equations of the XXX spin chain for arbitrary spin by analyzing the set of vacua of a certain supersymmetric theory. Particularly important is the formula~(\ref{1.5}) for the superpotential, since it is used in section~\ref{sec:AG} as a starting point to define the algebraic varieties studied in the present paper. These varieties come equipped with a toric action, defined in secs.~\ref{sec:pretor} and~\ref{sec:tor}, which can be used to define their equivariant (co)homology. Equivariant localization leads us to consider the fixed points of the torus, which are described in Prop.~\ref{prop:fp}. In sec.~\ref{sec:attsch} we define the attractors of the fixed points under a~$\CC^\times$-action, where $\CC^\times$ sits inside the torus. In sec.~\ref{sec:stabenv}, we propose a characterization of the stable basis, extending in a special case the concept of stable envelope of~\cite{MO-qg} to our setting. In general the existence of the stable basis is taken as a conjecture (Conjecture 2); however in Prop.~\ref{prop:stabcl} we do prove it in the case $n=2$ (two sites of the corresponding spin chain), which is relevant for the calculation of the $R$-matrix. We start sec.~\ref{sec:R} by defining the $R$-matrix through the action of the Weyl group $\mathcal S_n$ of $GL(n)$ on the stable classes. In sec.~\ref{sec:l2ex} we calculate the $R$-matrix for the case of spin $1$ ($\ell=2$), which provides the first non-trivial example beyond quiver varieties. We then analyze the general case in sec.~\ref{sec:general}, using the fact that all relevant varieties are in fact isomorphic to the ones at `half-filling' ($k=\ell$). 
The latter fact turns out to be a non-trivial lemma, proven in Appendix~\ref{app:isom}.

\section{Nekrasov--Shatashvili approach}\label{sec:NS}
\subsection{Supersymmetry in $\mathbb{R}^{1,1}$}
In the paper \cite{NS-short} the authors consider $\mathcal{N}=(2,2)$ supersymmetric (SUSY) gauge theories with matter fields in 2D Minkowski space $\mathbb{R}^{1,1}$. This theory has four real supercharges, and the corresponding coordinates in superspace will be denoted $\theta^+, \theta^-, \bar{\theta}^+, \bar{\theta}^-$. This type of SUSY is very much like $\mathcal{N}=1$ SUSY in 4D. In the latter case it is well-known that the integrals of $D$-terms of general superfields and the integrals of $F$-terms of chiral superfields are invariant under SUSY transformations. A chiral superfield $\psi$ in 4D satisfies $\bar{D}_a\psi=0$, where $a=1,2$ is a spinor ($SL(2,\CC)$) index. In 2D this condition corresponds to $\bar{D}_+\psi=\bar{D}_-\psi=0$, where $\pm$ are the light-cone indices. The two conditions transform independently under the Lorentz group in 2D, which is the multiplicative $\mathbb{R}_+$ diagonally embedded in $SL(2,\CC)$. Therefore in 2D, as compared to 4D, it is possible to introduce a different type of supersymmetric multiplet (field) $\Sigma$, called twisted chiral, characterized by the condition $D_+\Sigma=\bar{D}_-\Sigma=0$. This multiplet is particularly important, as the gauge field strength is part of such multiplet. Picking out $D$- and $F$-terms requires (Berezin)-integrating over the corresponding superspace coordinates, so the most general SUSY-invariant action has the form
\begin{eqnarray}\label{genaction}
\mathcal{S}=\int d^2x\!\!\!\!\!\!&&\left(\int\,d\theta^+ d\theta^- d\bar{\theta}^+ d\bar{\theta}^-\,\mathcal{K}(\psi, \bar{\psi}, \Sigma, \bar{\Sigma}, V)+\right.\\ \nonumber &&\left.+\int\,d\theta^+ d\theta^-\,\mathcal{F}(\psi)+\int\,d\bar{\theta}^+ d\theta^-\,\tilde{\mathcal{F}}(\Sigma)+c.c.\right)
\end{eqnarray}
Here $V$ is the gauge superfield, which encorporates the gauge connection $A_\mu$. $\Sigma$ is related to $V$ via $\Sigma=D_+\bar{D}_-V$. In \cite{NS-short} $V$ and $\Sigma$ are related to gauge fields, whereas $\psi, \bar{\psi}$ are the `matter fields', furnishing some representations of the gauge group. In~(\ref{genaction}) $\mathcal{K}$ is the generalized K\"ahler potential, and the function $\mathcal{F}$ is called the superpotential.

\subsection{K\"ahler quotients via supersymmetry}
$\mathcal{N}=(2,2)$ supersymmetry is very natural from the point of view of K\"ahler geometry. In this context the function $\mathcal{K}$ in the above action may be thought of as the K\"ahler potential on the target space of a $\sigma$-model. If the `matter' fields $\psi$ are coordinates on some K\"ahler manifold $M$, and the gauge fields $V$ have no kinetic terms (and if we set $\mathcal{F}=\tilde{\mathcal{F}}=0$), then `integrating out' these auxiliary gauge fields will lead to the target space $\M=M/\!/G$ -- the K\"ahler quotient of the original space w.r.t.\ $G$. Since the K\"ahler quotient formulations are known for flag manifolds, in this way one can write down SUSY actions for $\sigma$-models with such target spaces~\cite{Donagi}. For example, in the case of a Grassmannian $Gr(k, n)$ we may set $M=\CC^{kn}$ and $G=U(k)$.

\subsection{The $D-$ and $\mathcal{F}-$term conditions}
In general, if one has no kinetic terms for the gauge fields, and possibly a non-zero superpotential $\mathcal{F}$, the  action contains the following terms:
\begin{equation}\label{1}
\int d^2 x\,\left(\mathrm{Tr}(D \mu_\RR)+\|\nabla \mathcal{F}\|^2\right)\,.
\end{equation}
Here $\mu_\RR$ is a function of the fields, which is nothing but the (real) moment map for the action of $G$ on $M$. $D$ is an auxiliary field, which in this case serves as a Lagrange multiplier imposing the constraint
\begin{equation}\label{mueq1}
\mu_\RR=0\,.
\end{equation}
The second term in~(\ref{1}) is the ordinary potential energy of the chiral fields $\psi$. Restricting to the minimum of this potential, we obtain the \emph{complex} equations
\begin{equation}\label{feq1}
\nabla \mathcal{F}=0\,.
\end{equation}
\subsection{Twisted masses.}
In \cite{NS-short} the authors consider the low-energy effective theory corresponding to the models just introduced. As it turns out, this theory is closely related to the geometry of the algebraic variety defined by~(\ref{mueq1})-(\ref{feq1}). One of the requirements in~\cite{NS-short} is that the low-energy theory is an abelian gauge theory, with gauge group $U(1)^k$, the corresponding `field strength' twisted chiral fields being $\Sigma_1 \ldots \Sigma_k$. In order to eliminate all other fields, one needs to make sure they are massive. This is done by introducing the so-called `twisted masses'~\cite{AlvarezGaume} for the matter fields. The twisted masses are given by a matrix $m\in \mathfrak{h}_{\CC}$, where $\mathfrak{h}\subset \mathcal{G}$ is the Cartan subalgebra of the global symmetry algebra $\mathcal{G}$. For example, in the case of the Grassmannian $Gr(k, n)$ $\sigma$-model the global symmetry is $U(n)$ so there are $n$ complex twisted mass parameters. These will appear as equivariant parameters in the equivariant cohomology $H^\ast_{T^n}(Gr(k, n))$ (or in $H^\ast_{T^n}(\M)$ in the general case).

In principle, one could also introduce masses for $\psi$ by adding quadratic terms in the superpotential, $\mathcal{F}\sim m\,\psi^2$ (this generates mass terms in the Lagrangian via (\ref{1})). In general, however, it is not possible to do this in a gauge-invariant way. For instance, in the case of $Gr(k, n)$ one would need to write a holomorphic $U(k)$-invariant function $\mathcal{F}$ on $\CC^{kn}$, which does not exist. 

\subsection{The hyper-K\"ahler quotient.}
In certain situations, the equations~(\ref{mueq1}) and~(\ref{feq1}) may be rewritten as the conditions of vanishing of three moment maps, $\vec{\mu}=0$. This requires a non-zero, and very particular, superpotential $\mathcal{F}$ and gives rise to a hyper-K\"ahler quotient rather than a K\"ahler quotient. From the point of view of the $\sigma$-model this involves an enhancement of SUSY from $\mathcal{N}=(2,2)$ to $\mathcal{N}=(4,4)$. One prototypical situation where this happens is when we deal with the target space $\M=T^\ast(Gr(k,n))$, which can be rewritten as a hyper-K\"ahler quotient of flat space:
\begin{equation}
T^\ast(Gr(k,n))=\CC^{2kn}/\!/_{\text{HK}}\,U(k)
\end{equation}
Defining the action of $U(k)$ on $\CC^k$ in the standard way, we will parametrize the space $\CC^{2kn}$ by two matrices $Q\in \Hom(\CC^n, \CC^k), \tilde{Q}\in \Hom(\CC^k, \CC^n)$. The corresponding superpotential is
\begin{equation}
\mathcal{F}=\mathrm{Tr}(Q\tilde{Q}\Phi),
\end{equation}
where $\Phi\in \End(\CC^k)$ is a $k\times k$-matrix -- an additional field.  Assuming the K\"ahler potential $\mathcal{K}$ entering~(\ref{genaction}) is that of flat space, the real moment map featuring in~(\ref{1}) is
\begin{equation}\label{realmommap1}
\mu_\RR=QQ^\dagger-\Qt^\dagger\Qt+[\Phi, \Phi^\dagger]+r\,\mathbf{1}_k\,=0,
\end{equation}
$r$ being the Fayet--Iliopoulos parameter. Throughout this paper we will assume $r>0$. The equations $\nabla \mathcal F=0$ are now
\begin{equation}
\Qt\Phi=0,\quad\quad \Phi Q=0,\quad\quad \mu_{\CC}:=Q\Qt=0\,.
\end{equation}
As the notation suggests, the last equation is the complex moment map relevant for the hyper-K\"ahler quotient. It may be shown that the equations imply $\Phi=0$, so the matrix field $\Phi$ effectively drops out (a more general case is discussed below in Sec.~\ref{higherspingen}).

As a side comment, we note that one could get rid of the fields $Q$ and $\Phi$ by introducing twisted masses for these fields (called $\varepsilon$ and $\varphi$ below) and then sending these masses to infinity. From the point of view of equivariant cohomology, this corresponds to the reduction from $H^\ast_{T^n}(T^\ast Gr(k, n))$ to $H^\ast_{T^n}(Gr(k, n))$, and at the level of the $\sigma$-model one obtains an asymptotically free theory from a conformal one. This is quite parallel to the constructions in four-dimensional gauge theory~\cite{GMN}, where conformal models are the basic ones, from which asymptotically free theories may be obtained in analogous limits.

\subsection{The effective action and Bethe equations.}
Let us now return to the construction of the low-energy theory, once the twisted masses are turned on. Integrating out massive fields produces an effective action for $\Sigma$, which is of the form~(\ref{genaction}) with $\mathcal{F}_{\mathrm{eff}}=0$ and $\mathcal{K}_{\mathrm{eff}}=\mathcal{K}_{\mathrm{eff}}(\Sigma, \bar{\Sigma})$.
If one eliminates the supercoordinates (by performing Berezin integration), one will obtain a term in the effective action of the form $\int\,d^2x\,\|\nabla \tilde{\mathcal{F}}_{\mathrm{eff}}\|^2$, so in the discussion of the low-energy theory one may, once again, restrict to
\begin{equation}
\nabla \tilde{\mathcal{F}}_{\mathrm{eff}}=0\,.
\end{equation}
Note that, due to $\mathcal{N}=({2,2})$ SUSY, there is only a one-loop contribution to $\tilde{\mathcal{F}}_{\mathrm{eff}}$, which can be calculated exactly. The claim of \cite{NS-short} is that $\tilde{\mathcal{F}}_{\mathrm{eff}}$ is the Yang--Yang functional corresponding to the $\mathit{XXX}$ spin-${1\over 2}$ chain of length $n$ with $k$ excitations, and $\nabla \tilde{\mathcal{F}}_{\mathrm{eff}}=0$ are the corresponding Bethe equations. The interpretation is that the solutions should be in one-to-one correspondence with the supersymmetric vacua of the original system.

\subsection{Higher spin generalizations.}\label{higherspingen}
One may now invert the argument and ask, what field theoretic (geometric) situation would lead to an $\mathit{XXX}$ spin-$s$ chain. The answer given in~\cite{NS-short} is that one should take a different superpotential:
\begin{equation}\label{1.5}
\mathcal{F}=\mathrm{Tr}(Q\Qt\Phi^{\ell})\,.
\end{equation}
where $\ell$ is related to the spin as $\ell=2s$. This leads to the equations
\begin{align}\label{2}
\Qt\Phi^{\ell}&=0,
\\\label{2.5}
\Phi^{\ell} Q&=0,
\\
\label{3} \sum\limits_{i=0}^{\ell-1}\,\Phi^i Q\Qt \Phi^{\ell-i-1}&=0\,.
\end{align}
One still has the same real moment map equation~(\ref{realmommap1}). One should think of this generalized situation as follows: one has a complex algebraic set $M$ defined by the equations \eqref{2}--\eqref{3} inside $\CC^{2kn}\oplus \CC^{k^2}$ (the latter summand being related to the matrix $\Phi$), and we are taking its K\"ahler quotient w.r.t.\ $U(k)$. From this perspective the real moment map equation determines the `stable set' inside $M$.

We derive an important consequence of the above equations. The equations (\ref{2}--\ref{2.5}) and \eqref{realmommap1} above imply
\bea\label{nilpotent}\Phi^\ell=0
\eea
for $r\neq 0$. 
To see this, first we multiply (\ref{realmommap1}) by $[\Phi^\ell, (\Phi^\dagger)^\ell ]$ and take the trace. Using (\ref{2}), as a result we get
\bea
\tr((\Phi^\dagger)^\ell Q Q^\dagger \Phi^\ell)+\tr(\Phi^\ell \Qt^\dagger \Qt (\Phi^\dagger)^\ell)+\tr([\Phi^\ell, (\Phi^\dagger)^\ell ][\Phi, \Phi^\dagger])=0\,.
\eea
Swapping the commutator twice in the last term, we can show that $$\tr([\Phi^\ell, (\Phi^\dagger)^\ell ][\Phi, \Phi^\dagger])=\tr([\Phi, (\Phi^\dagger)^\ell][\Phi^\ell, \Phi^\dagger]).$$ As a result, the above equation takes the form
\bea
\|A\|^2+\|B\|^2+\|C\|^2=0\,,
\eea
where
\bea\label{ABCzero}
A=Q^\dagger \Phi^\ell,\quad\quad B=\Phi^\ell \Qt^\dagger,\quad\quad C=[\Phi^\ell, \Phi^\dagger]\,.
\eea
This shows that $A=B=C=0$. We now multiply (\ref{realmommap1}) by $\Phi^\ell (\Phi^\dagger)^\ell$ and take the trace. Using (\ref{2}) and (\ref{ABCzero}), we arrive at $r\,\tr(\Phi^\ell (\Phi^\dagger)^\ell)=0$, implying $\Phi^\ell=0$ for $r\neq 0$.

\subsection{The torus of global symmetries.}\label{sec:pretor}
We now answer the following question: what is the most general torus leaving the algebraic set defined by (\ref{2})-(\ref{3})  invariant? Clearly, one still has the Cartan torus $(\CC^\times)^{n}$ of $GL(n, \CC)$. However, one has two additional $\CC^\times$-actions:
\bea
(\CC^\times)_1:\quad Q\to e^\varepsilon\,Q,\quad\quad (\CC^\times)_2:\quad \Phi\to e^\varphi\,\Phi\,.
\eea
(If one takes into account the real moment map equation as well, one should restrict to the unitary subgroups.) 
Notice, however, that this transformation does not, in general, preserve the superpotential $\mathcal{F}$, as under this transformation $\mathcal{F}\to e^{\varepsilon+ \ell \,\varphi} \mathcal{F}$, so in order for the action~(\ref{genaction}) to be invariant, one would need to rescale the supercoordinates $\theta$, i.e., perform a compensating $R$-symmetry transformation. This leads to an obstacle in introducing a twisted mass for the transformation, under which the superpotential is non-trivially scaled. To see this, let us recall that twisted masses may be introduced by dimensionally reducing a 4D theory of chiral superfields on $S^1\times S^1$ equivariantly w.r.t.\ some abelian global symmetry. In the simplest example, imagine one has, in a 4D theory, a single complex scalar field $\phi(x_0, x_1, x_2, x_3)$ and a Weyl fermion $\chi(x_0, x_1, x_2, x_3)$, and the theory itself has charge symmetry $\phi\to e^{i q_1 \alpha} \phi$, $\chi \to e^{i q_2 \alpha} \chi$ ($q_1, q_2 \in \mathbb{Z}$ being the charges). The equivariant dimensional reduction in question along the directions $x_2, x_3$ amounts to assuming the following dependence of $\phi, \chi$ on these coordinates:
\bea\label{equivdimred}
\phi=e^{i q_1 (m_1 x_2+m_2 x_3)} \tilde{\phi}(x_0, x_1),\quad\quad \chi=e^{i q_2 (m_1 x_2+m_2 x_3)} \tilde{\chi}(x_0, x_1)
\eea
In this case $m:=m_1+i m_2\in \CC$ is the twisted mass for the respective $U(1)$ symmetry. In fact, the applicability of~(\ref{equivdimred}) makes no use of the fact that the original 4D theory was supersymmetric. Let us now assume that $\phi$ and $\chi$ are the two components of a chiral superfield $\boldsymbol{\phi}$ in a SUSY 4D theory. In other words,
\bea\label{superfield}
\boldsymbol{\phi}=\phi+\theta^a\,\chi_a+\ldots\,,
\eea
where $\ldots$ denotes higher terms in the $\theta$-expansion. The R-symmetry, by definition, acts non-trivially on the superspace coordinates, $\theta^a \to e^{i\beta}\,\theta^a$. As a result (since the superfield $\boldsymbol{\phi}$ transforms homogeneously) the $R$-charges of $\phi$ and $\chi$ are unequal ($q_1\neq q_2$). In this case, however, the Ansatz (\ref{equivdimred}) is not compatible with the superfield structure (\ref{superfield}). One can rephrase this by recalling that the superfield structure leads to the SUSY transformation $\delta \phi\sim \chi$, which again is not compatible with (\ref{equivdimred}) for $q_1\neq q_2$.


The conclusion is that the twisted masses may be introduced for the torus
\bea\label{eq:torus}
T:=T_0\times \CC^\times\,,
\eea
where $T_0\simeq (\CC^\times)^n$ is the Cartan torus of $GL(n, \CC)$, and the second factor is the subgroup $\CC^\times\subset (\CC^\times)_1\times (\CC^\times)_2$ specified by the choice
\bea
\varepsilon+ \ell \,\varphi=0\,.
\eea

\section{Algebro-geometric construction}\label{sec:AG}
We now reformulate the ideas of the previous section in a (complex) algebro-geometric setting.

\subsection{The varieties}\label{sec:var}
As before, we fix three integers $n,k\ge 0$, $\ell\ge1$.
We work over the field of complex numbers. Let $V$ (resp.\ $W$) be a vector space
of dimension $k$ (resp.\ $n$).

We consider the function $\mathcal F$ on $\Hom(W,V)\oplus
\Hom(V,W)\oplus\End(V)$ given by
\begin{equation}\label{eq:defF}
\mathcal F(Q,\Qt,\Phi)=\tr(Q \Qt \Phi^\ell)
\end{equation}
and its critical locus
\begin{equation}\label{eq:defpreM}
M(k,n,\ell)=\left\{(Q,\Qt,\Phi): d\mathcal{F}=0\right\}
\end{equation}
The explicit equations for the critical locus were already given above: \eqref{2}--\eqref{3}.

$GL(V)$ acts on $M$ in the natural way, so that we may define the {\em GIT quotient}
\begin{equation}\label{eq:defM}
\M(k,n,\ell)=M(k,n,\ell)/\!/GL(V)
\end{equation}
Note that we have substituted the nonholomorphic equation \eqref{realmommap1} with taking the quotient
by the complexified gauge group $GL(V)\cong GL(k)$. The equivalence of these two operations is known in the setting
of smooth, projective varieties (see \cite{Kirwan-book}); here, we just take \eqref{eq:defM} as our
definition of $\M(k,n,\ell)$.

When there is no risk of confusion, we shall suppress the dependence on $k,n,\ell$ and simply
write $\M=\M(k,n,\ell)$, etc.

The semi-stable set is given by the condition
\begin{equation}\label{eq:stab}
\rank
\begin{pmatrix}\Qt\\\Qt\Phi\\\vdots\\\Qt\Phi^{\ell-1}
\end{pmatrix}
=k
\end{equation}
Given $(Q,\Qt,\Phi)$ in the stable set, its stabilizer in $GL(V)$
can be computed by noting that the columns of $\Phi^i Q$, $i=0,\ldots,\ell-1$ span $\CC^k$, so that
any element of $GL(V)$ acts by left multiplication (if it stabilizes $\Phi$) and leaves them invariant (if it stabilizes
$\Qt$) must be equal to the identity.
Furthermore, by the same token, an orbit is uniquely characterized by the $k\times k$ minors of the matrix in \eqref{eq:stab}
(which fixes $\Phi$ and $\Qt$ up to $GL(V)$ action) and by the $\Qt \Phi^i Q$ (which fixes $Q$). More details
about this parameterization will be given in Sect.~\ref{sec:param}. In any case, we conclude that all orbits
are closed and so the stable set equals the semi-stable set. We denote it with the superscript $s$.
\rem[gray]{that has implications for singularities (only orbifold?). what about (co)homology?
rationally smooth?}

The same argument also implies that the stability condition \eqref{eq:stab} combined with \eqref{2} implies,
just like the real moment map equation \eqref{realmommap1}, that $\Phi^\ell=0$, so that we may equivalently
write
\begin{equation}\label{eq:defM2}
\M=\left\{(Q,\Qt,\Phi):
\Phi^\ell=0,\ 
\sum_{i=0}^{\ell-1}\Phi^i Q\Qt \Phi^{\ell-1-i}=0
\right\}/\!/ GL(V)
\end{equation}
$\M$ 
is in general neither reduced nor irreducible. Both problems can be cured
by observing the following:
the two equations defining $\M$ can be recombined as
\[
(\Phi+\epsilon Q\Qt)^\ell=0\pmod {\epsilon^2}
\]
so that if we define
\begin{equation}
\mathcal{O}=\{\Phi\in \End V: \Phi^\ell=0\}
\end{equation}
then
\begin{equation}
\M=
\left\{(Q,\Qt,\Phi):
\Phi\in \mathcal{O},\ Q\Qt \in T_\Phi \mathcal{O}
\right\}/\!/ GL(V)
\end{equation}
Set-theoretically,
$\mathcal{O}$ is a union of $GL(V)$ orbits inside $\End V$. Each orbit is given by fixing the Jordan
form of $\Phi$, which can be described by a
partition $\lambda=(\lambda_1\ge\lambda_2\ge\cdots)$, where the $\lambda_i$ are the sizes of the 
Jordan blocks of $\Phi$
(satisfying $1\le\lambda_i\le \ell$ and $\sum_i \lambda_i=k$).
The open orbit, denoted $\mathcal{O}_1$, 
corresponds to the largest partition in dominance order, namely
$\lambda=(\ell^q r)$,
writing $k=q\ell+r$, $0\le r<\ell$.
Consider $\mathcal{O}_1$ with its reduced structure, so that in particular it is smooth,
and then define
\begin{equation}\label{eq:defM1o}
\M^o_1=
\left\{(Q,\Qt,\Phi):
\Phi\in \mathcal{O}_1,\ Q\Qt \in T_\Phi \mathcal{O}_1
\right\}^s/ GL(V)
\end{equation}
and
\begin{equation}
\M_1=\overline{\M^o_1}\subset \M
\end{equation}
$\M_1$ (with its reduced structure) is a subvariety of $\M$, as will be proved in the next section.

\begin{rmk}
There is a natural map from $\M(k,n,\ell)$ to $\M(k,\ell n,1)\cong T^*Gr(k,\ell n)$ 
that associates to $(Q,\Qt,\Phi)$
\[
Q^{(1)}=\begin{pmatrix}\Phi^{\ell-1}Q&\ldots& Q
\end{pmatrix}
\qquad
\Qt^{(1)}=\begin{pmatrix}\Qt\\\Qt\Phi\\\vdots\\\Qt\Phi^{\ell-1}
\end{pmatrix}
\]
($\Phi=0$ at $\ell=1$). This is a first step towards a geometric interpretation of the fusion procedure
mentioned in the introduction.
\end{rmk}

\subsection{Poisson / symplectic structure}\label{sec:poisson}
We want to define a Poisson structure on $\M_1$.
We start with the following lemma:
\begin{lem}\label{lem:stab}
Given $\Phi\in \mathcal O_1$, the commutant $\mathrm{stab}(\Phi):=\Ker \mathrm{ad}_{\Phi}$ is
generated by $\sum\limits_{i+j=\ell-1} \Phi^i X \Phi^j$, $X\in \End V$, and $\Phi^i$, $i=0,\ldots,\ell-2$.
\end{lem}
\begin{proof}
We define the operator $\kappa(X):=\sum\limits_{i+j=\ell-1} \Phi^i X \Phi^j$. It is easy to see that $\kappa(X)\in \mathrm{stab}(\Phi)$. To characterize $\mathrm{stab}(\Phi)\setminus \mathrm{Im}(\kappa)$, we recall a more direct description of $\mathrm{stab}(\Phi)$. Since $\Phi\in \mathcal O_1$, the vector space $V$ may be decomposed into $q$ Jordan towers of length $\ell$, with generating vectors $v_1, \ldots, v_q$ and a remaining smaller tower of length $r<\ell$, with generating vector~$v_0$. We set $W_0=\mathrm{Span}(v_1, \ldots v_q)$,  $V_0=W_0\oplus \CC v_0$ and accordingly $W_k=\Phi^k(W_0)$, $V_k=W_k\oplus \CC \Phi^k(v_0)$ for $k=1, \ldots, \ell-1$. Note that $W_k\simeq W_0$ for all $k$, and $V_k\simeq V_0$ for $k<r$.
Clearly, $V=\oplus_{k=0}^{\ell-1} V_k$, and we define the projection operator $\pi_k: V\to V_k$.

An operator $B\in \mathrm{stab}(\Phi)$ is defined by its action on $V_0$. More exactly, in general $B$ sends $W_0\to V$, and $v_0$ to $\oplus_{k=0}^{r-1} \Phi^k (v_0) \oplus_{l\geq \ell-r} W_l$. Let us show that $B\in \mathrm{Im}(\kappa)$ if and only if $B(v_0)\in \oplus_{l\geq \ell-r} V_l=\mathrm{Im}(\Phi^{\ell-r})$. First, from the definition of $\kappa$ and from the fact that $\Phi^l (v_0)=0$ for $l\geq r$ it is clear that $\kappa(X)(v_0)\in \mathrm{Im}(\Phi^{\ell-r})$.

Now assume the converse, i.e., that $B(v_0)\in \oplus_{l\geq \ell-r} V_l$. Let us restrict $B$ to $V_0$ and consider $\pi_k\circ B\big|_{V_0}=\Phi^k B_k$, where $B_k\in \mathrm{End}(V_0)$. We will now consider an operator $X$ that maps $V_k\to V_0$ for all $k$. In this case $\pi_k\circ \kappa(X)=\Phi^k X \Phi^{\ell-1-k}$. If the operator $X$ additionally satisfies $B_k=X \Phi^{\ell-1-k}$ for all $k$, then $B=\kappa(X)$, as claimed. Therefore the question is whether $B_k\in \mathrm{End}(V_0)$ may be factored through the map $\Phi^{\ell-1-k}\big|_{V_0}$. As remarked earlier, for $\ell-1-k<r$ this map is an isomorphism $V_0\simeq  V_{\ell-1-k}$, so there is no obstruction.
On the other hand, for $\ell-k> r$, $\Phi^{\ell-1-k}$ annihilates $v_0$, and is an isomorphism between $W_0\simeq  W_{\ell-1-k}$. However, by the assumption $B(v_0)\in \oplus_{l\geq \ell-r} V_l$, so that we may set $B_k(v_0)=0$ for $k<\ell-r$. Therefore again we can find $X$ satisfying $B_k=X \Phi^{\ell-1-k}$. Therefore in this case $B=\kappa(X)$.

In general, $B(v_0)=\sum\limits_{k=0}^{r-1} \,c_k\, \Phi^k (v_0)\quad(\mathrm{mod}\quad \mathrm{Im}(\Phi^{\ell-r}))$, so that there is always a combination $\widetilde{B}=B-\sum\limits_{k=0}^{r-1} \,c_k\, \Phi^k$ such that $\widetilde{B}(v_0)\in \mathrm{Im}(\Phi^{\ell-r})$ and hence $\widetilde{B}\in \mathrm{Im}(\kappa)$. Since $r<\ell$, this completes the proof.
\end{proof}

We now claim that $\M_1$ can be obtained by symplectic reduction as follows.
Consider the extended space $T^* \Hom(V,W)\oplus T^* \End(V)$, or
\[
\left\{ (Q,\Qt,\Phi,\alpha):\ Q\in \Hom(W,V),\ \Qt\in\Hom(V,W),\ \Phi\in \End(V),\ \alpha\in \End(V)
\right\}
\]
with its symplectic form
\[
\Omega=\tr(dQ\wedge d\Qt)+\tr(d\Phi\wedge d\alpha)
\]
Then the natural action of $GL(V)$ on this space is given by the moment map
\begin{equation}\label{eq:newmm}
\mu_{\mathfrak{gl}(V)}=Q\Qt+[\Phi,\alpha]
\end{equation}
We also consider the family of Poisson--commuting Hamiltonians given by the entries of $\Phi^\ell$, and
by $\tr \Phi^i$, $i=1,\ldots,\ell-1$. We consider the symplectic reduction w.r.t.\ both, i.e., by
$\End(V) \rtimes GL(V)$. For the $GL(V)$ part we use the same GIT quotient as for $\M$.
We call the result of this symplectic reduction $\M'$. ($\M'$ will not be needed outside of this section.)

We first take the quotient by the additive group $\End(V)$. The effect is to impose the equations
$\Phi^\ell=0$, $\tr \Phi^i=0$, $i=1,\ldots,\ell-1$, and to take the quotient by the action by translation
on $\alpha$ given by differentiating these equations. In the first step, we restrict our space
to $\Hom(W,V)\times \Hom(V,W)\times \mathcal O \times \End(V)$; since $\mathcal O$ is normal
we can further restrict to $\mathcal O_1$ \cite[Sect.~8]{Jantzen}. Because
restriction to $\mathcal O_1$ and taking the $\End(V)$-invariant space of functions obviously commute,
we can now perform the second step,
and recognize exactly the generators of $\text{stab}(\Phi)$
given in Lem.~\ref{lem:stab}, so that the quotient amounts to having $\alpha$ live in $\End(V)/\text{stab}(\Phi)$.

At this stage we can take the quotient by the $GL(V)$ action. We continue the analysis of $\Phi\in \mathcal O_1$,
noting that the equations that follow do depend on $\Phi$, so contrary to the reasoning of the previous paragraph
this is a nontrivial restriction (we shall go back to the general case below).
Imposing $\mu_{\mathfrak{gl}(V)}=0$ in
\eqref{eq:newmm}, i.e.,
\begin{equation}\label{eq:mm}
Q\Qt = [\alpha,\Phi]
\end{equation}
we find that $\alpha\in \End(V)/\text{stab}(\Phi)$ can be eliminated altogether.
Now \eqref{eq:mm} is equivalent to
$Q\Qt \in T_{\Phi}\mathcal O_1$. In fact, we can pick a representative $\Phi_0\in \mathcal O_1$ and simply
write
\begin{equation}
\M^o_1=
\left\{(Q,\Qt):
Q\Qt \in T_{\Phi_0} \mathcal{O}_1
\right\}^s/\, \text{Stab}(\Phi_0)
\end{equation}
where $\text{Stab}(\Phi)=\Ker \text{Ad}_{\Phi}$.

By the same argument that shows that the stabilizer of any element in $M^s$ has trivial stabilizer in $GL(V)$,
we find here that any $(Q,\Qt)$ has trivial stabilizer in $\text{Stab}(\Phi_0)$,
i.e., the whole of the stable set is regular for the action of $\text{Stab}(\Phi_0)$.

We conclude that $\M_1^o$ is a {\em smooth}, {\em symplectic}\/ open subvariety of $\M'$ 
(in analytic terms, it is endowed with a holomorphic symplectic form); 
so that its closure $\M_1$ is also irreducible. 

We also derive a dimension formula:
\begin{lem}
Recalling $k=q\ell+r$, $0\le r<\ell$,
\begin{equation}\label{eq:dim}
\dim \M_1 = 2kn - 2 \dim \mathrm{Stab}(\Phi)=2(kn-\ell q^2-r(2q+1))\,.
\end{equation}
\end{lem}
\begin{proof}
One can deduce the dimension of the stabilizer from a general formula (see \cite[Sec.~1.3]{Humph-conjclass} or \cite[Sec.~3.1]{Jantzen}). If the Jordan form of $\Phi$ is given by a partition of type $\lambda_1\geq \ldots \geq \lambda_p\geq0$, then
\bea
\mathrm{dim} \,\mathrm{Stab}(\Phi)=\sum\limits_{i=1}^p\,(2i-1)\,\lambda_i\,.
\eea
In our case $p=q+1$, $\lambda_i=\ell$ for $1\leq i \leq q$ and $\lambda_{q+1}=r$, giving $\mathrm{dim} \,\mathrm{Stab}(\Phi)=\ell q^2+r(2q+1)$.

As a consistency check let us show that the dimension~(\ref{eq:dim}) is non-negative, if the stable set is non-empty. The stability condition (\ref{eq:stab}) implies that the number of Jordan blocks in $\Phi$ (i.e. the number of non-zero $\lambda_i$'s) does not exceed $n$. Indeed, this is a simple consequence of the fact that the matrix in~(\ref{eq:stab}) may be viewed as a collection of  $n$ Jordan towers of vectors in $\CC^k$, and by the stability requirement the vectors in these towers span $\CC^k$. Since some of these vectors may be linearly dependent, the actual number of Jordan blocks may be smaller. On the other hand, in the case, which is considered in the lemma, the number of Jordan blocks is $q$ for $r=0$ and $q+1$ for $\ell>r>0$. Therefore from (\ref{eq:dim}) we have $\dim \M_1\geq 0$ for $r=0$ and $\dim \M_1\geq 2((q \ell+r)(q+1)-\ell q^2-r(2q+1))=2q(\ell-r)>0$ for $r\neq 0$.

\end{proof}

We now go back to $\M'$, i.e., without restriction on
the Jordan form of $\Phi$. The same dimension calculation can be easily performed
and results in a dimension that is less or equal to that of $\M_1$;
so that we conclude that $\M_1$ is an irreducible component of $\M'$.
(for $n>2$ it is not true in general that $\M'$ is irreducible and equal to $\M_1$.)
In particular, $\M_1$, just like $\M'$, is Poisson.
$\M_1$ typically has singularities outside $\M_1^o$. It is an interesting question
whether it admits a symplectic resolution.

\subsection{Parametrization}\label{sec:param}
In what follows, explicit computations will require introducing coordinates on $\M$.
Pick bases of $V$ and $W$, so that $Q$, $\Qt$, $\Phi$
are identified with their matrices.
Elementary invariant theory implies that the subring of $SL(V)$ invariant functions on $M$ 
with nonnegative $GL(1)\subset GL(V)$ weight (which we choose as our projective degree) has two types of generators:
\begin{itemize}
\item Degree $0$ generators which are the entries $u_{i,a,b}$ of the matrices $u_i=\Qt \Phi^i Q$, $i=0,\ldots,k-1$,
$a=1,\ldots,n$, $b=1,\ldots,n$;
\item Degree $1$, Pl\"ucker-like generators $x_s$, where
each $s$ is a subset of $\{0,\ldots,k-1\}\times\{1,\ldots,n\}$ of cardinality $k$;
$x_s$ is given by the determinant of the columns $s_i=\{j: (i,j)\in s \}$ of $\Qt\Phi^i$, $i=0,\ldots,k-1$.
When we write these coordinates more explicitly for small values of $n$, we shall use the notation $x_s=x_{s_0\ldots s_{k-1}}$ where the entries of
each $s_i$ are stacked vertically, e.g., $x_{\{(0,1),(1,1),(1,2)\}}=x_{\begin{smallmatrix}1&1\\&2  \end{smallmatrix}}$.
\end{itemize}
In other words,
the $x_s$ form all the $k\times k$ minors of the matrix
\[
\At=\begin{pmatrix}\Qt\\\Qt\Phi\\\vdots\\\Qt\Phi^{\ell-1}
\end{pmatrix}
\]
and the semi-stability condition \eqref{eq:stab}, namely, $\rank \At=k$, 
simply requires that the $x_s$ should not all vanish simultaneously.

We shall not need to write in full the equations satisfied by the $u_i$ and $x_s$ in $\M$ (or $\M_1$). We mention
some of the equations of $\M$ and $\M_1$: \rem[gray]{bit of a mess: there's $\M$, its radical, $\M'$
obtained by imposing $\tr u_i=0$, and $\M_1$... all of them are different}
\begin{itemize}
\item The $x_s$ satisfy linear and (Pl\"ucker-like) quadratic equations.
\item The $u_i$ and $x_s$ satisfy together bilinear equations, some of
which we shall need.
Fix $i\in \{1,\ldots,n\}$. Consider the matrix $\At$ above 
and add to it one more column, which is a linear combination of the first $k$
with coefficients $Q_{r,i}$, $r=1,\ldots,k$, $i$ fixed. Expressing that the minors of size $k+1$ of this matrix are
zero, we obtain
\begin{equation}\label{eq:bil}
\sum_{r=1}^{k+1} (-1)^r x_{s\backslash s_r} u_{s_{r,1},s_{r,2},i}=0
\end{equation}
The identity holds for all $(k+1)$-subsets $s$ of $\{0,\ldots,\ell-1\}\times \{1,\ldots,n\}$ and all $i=1,\ldots,n$.
\item Finally, there are equations involving only the $u_i$:
quadratic equations which can be easily deduced from the equations of $M$:
\begin{equation}\label{eq:aff}
\sum_{i=1}^{j} u_{\ell-i} u_{\ell-(j+1-i)} = 0 \qquad j=1,\ldots,\ell
\end{equation}
as well as some rank equations that we shall not need to write.
These form the equations of the affinization $\Spec \Gamma(\mathcal O_\M)$ of $\M$.
\end{itemize}
It is not clear in general whether these equations generate the (saturated) ideal of $\M$, 
or if a similar description can be made for
$\M_1$, and we do not pursue this further.
The reader may want to check that this description {\em is}\/ complete 
in the case $\ell=1$, when $\M(k,n,1)=\M_1(k,n,1)$ is simply $T^*Gr(k,n)$.

$\M$ or $\M_1$ are projective over their affinization, but only quasi-projective (over $\CC$).
In what follows, we shall several times encounter the {\em projective}\/ variety
\begin{align}\label{eq:defP}
\P(k,n,\ell)&=\M(k,n,\ell) \cap \{ u_i=0,\ i=0,\ldots\}
\\\notag
&= \{ (\Qt,\Phi): \tr \Phi^i=0,\ i=1,\ldots,\ell-1,\ \Phi^\ell=0 \} /\!/ GL(V)
\end{align}
(In the case $\ell=1$, we have $\P(k,n,1)\cong Gr(k,n)$.)

Note that for $\P$, there is no distinction analogous to $\M$ vs $\M_1$:
first $\P$ is by definition reduced irreducible; secondly,
since in the definition \eqref{eq:defP} $\Qt$ does not satisfy
any equations (and the stability condition is an open condition), 
the set of orbits $(\Qt,\Phi)$ where $\Phi\in \mathcal O_1$
is dense in $\P$. In particular,
\begin{equation}\label{eq:PM1}
\P \subset \M_1
\end{equation}
There is also a projection $p: \M \to \P$ that consists in forgetting $Q$, or after the quotient by $GL(V)$,
in forgetting the $u_i$ variables (thus retaining only the $x_s$).
\rem[gray]{does this imply a relation between $H(\M)$ and $H(\P)$?
ask AK about co vs homology in noncompact case}

\subsection{Torus action, equivariant homology}\label{sec:tor}
$GL(W)$ naturally acts on $M$, and it descends to an action on $\M$ and restricts to one on $\M_1$ (the action factors through
$PGL(W)$, but it is convenient to include the trivial $GL(1)$ action).
Similarly, $\CC^\times\times \CC^\times$ acts by scaling $Q$ and $\Phi$ on $M$,
and so on $\M$ or $\M_1$.

We now fix a basis of $W\cong \CC^n$, and define the torus
\begin{equation}\label{eq:torus2}
\hat T=T_0 \times \CC^{\times} \times \CC^{\times}
\end{equation}
where the first factor is the Cartan torus $T_0\cong (\CC^\times)^n$ of $GL(W)\cong GL(n)$, 
and the two $\CC^\times\times \CC^\times$ are as above.
As already discussed in \ref{sec:pretor},
we are really interested in a subtorus $T$ of codimension $1$
inside $\hat T$, cf \eqref{eq:torus}. Define 
the obvious coordinates (weights) $z_1,\ldots,z_n,\varphi,\varepsilon$ on the Lie algebra of the torus $\hat T$ \eqref{eq:torus2}.
Then the Lie algebra of $T$, as a Lie subalgebra of that of $\hat T$, is given 
by setting
\begin{equation}\label{eq:subtorus}
\varepsilon=-\ell\varphi
\end{equation}
$\hat T$ remains natural in the sense that all our constructions are $\hat T$-equivariant.
(Also, note that the definition of $\hat T$ does not depend on $\ell$, whereas that of $T$
does. This will be discussed again in Sect.~\ref{sec:RR})
\rem[gray]{note that the weight of the symplectic form, which [MO] call $\hbar$, is $-\varepsilon=\varphi/\ell$.
$\varphi$ would be the $\hbar$ in the standard convention of integrable models, which disagrees with [MO] for $\ell>1$.
also note the $\varphi$ action {\em is}\/ symplectic.}

\begin{rmk}
The weight of the symplectic form, denoted $\hbar$ in \cite{MO-qg}, is here $-\varepsilon$.
Only at $\ell=1$ does it coincide with $\varphi$, the natural homogeneity parameter for spectral parameters
(in the language of integrable systems).
\end{rmk}

Because $\M_1$ is singular, the correct way to study its cohomology in the
context of quantum integrability is probably to consider ``vanishing cycles'' as in \cite{YZ-vanish}.
However, for the purposes of this paper we need very little structure and so we use instead
the following low-brow approach. We consider the {\em Borel--Moore equivariant homology of $\M_1$}, denoted
$H_*^{\hat T}(\M_1)$. 
We consider it as a module over the equivariant cohomology ring of a point 
$H_{\hat T}(\cdot)\cong \ZZ[z_1,\ldots,z_n,\varphi,\varepsilon]$.
The map to $H_*^T(\M_1)$ is given once again by \eqref{eq:subtorus}.

Because $\M_1$ is singular, it is not known whether $H_*^{\hat T}(\M_1)$ is a {\em free}\/ module
over $H_{\hat T}(\cdot)$ (this is closely related to the question, which we shall not address here,
whether $\M_1$ is {\em equivariantly formal}\/);
however this will not concern us, because we shall resort to equivariant localization. Denoting $\M_1^{\hat T}$ the set of $\hat T$-fixed points of $\M_1$, recall
\begin{lem*}[see e.g.~{\cite[Lemma 1]{Brion-poincare}}]
After inverting finitely many weights in $H_{\hat T}(\cdot)$, the embedding map
$H^{\hat T}_*(\M_1^{\hat T})\to H^{\hat T}_*(\M_1)$ is an isomorphism.
\end{lem*}
The same statement holds with $\hat T$ replaced with $T$.

We now proceed to compute the set of fixed points:
\begin{prop}\label{prop:fp}
The set of torus fixed points of $\M_1$ is indexed
by integer sequences in $
\{0,\ldots,\ell\}^n$ whose sum of entries is $k$:
\[
\M_1^T = \M_1^{\hat T}=\left\{ p_{k_1,\ldots,k_n},\ 
0\le k_i\le\ell,\ \sum_{i=1}^n k_i=k\right\}
\]
where the fixed point $p_{k_1,\ldots,k_n}$ associated to the sequence $(k_1,\ldots,k_n)$
can be described in the following two equivalent ways:
\begin{itemize}
\item A representative of $p_{k_1,\ldots,k_n}$ in $M^s$ is given by
\begin{align}
\Phi&=\text{Jordan normal form with blocks of size $k_i$ (ignoring zero size blocks), i.e.,}
\label{eq:repPhi}
\\\notag
&=\begin{tikzpicture}[
baseline=-\the\dimexpr\fontdimen22\textfont2\relax]
\matrix[matrix of math nodes,left delimiter=(,right delimiter=),ampersand replacement=\&] (m){
0\&1\\
 \&0\&1\\
 \& \&0\\
\&\&\&0\&1\\
\&\&\& \&0\\[-2mm]
\&\&\& \&\&\ddots\\
};
\node[rectangle,draw,inner sep=-1pt,fit=(m-1-1) (m-3-3)] (a) {};
\node[rectangle,draw,inner sep=-1pt,fit=(m-4-4) (m-5-5)] (b) {};
\draw[latex-latex] ([yshift=3pt]a.north west) -- node[above] {$\ss k_1$} ([yshift=3pt]a.north east);
\draw[latex-latex] ([yshift=3pt]b.north west) -- node[above] {$\ss k_2$} ([yshift=3pt]b.north east);
\end{tikzpicture}
\\
\Qt_{ij}&=\begin{cases}
1& \text{if }k_i>0\text{ and } j=k_0+\cdots+k_{i-1}+1
\\
0&\text{else}
\end{cases}, \text{ i.e.},
\\\notag
\Qt
&=\begin{tikzpicture}[
baseline=-\the\dimexpr\fontdimen22\textfont2\relax]
\matrix[matrix of math nodes,left delimiter=(,right delimiter=),ampersand replacement=\&] (m){
1\&0\&0\\
\&\&\&1\&0\\[-2mm]
\&\&\&\&\&\ddots\\
};
\node[rectangle,draw,inner sep=-1pt,fit=(m-1-1) (m-1-3)] (a) {};
\node[rectangle,draw,inner sep=-1pt,fit=(m-2-4) (m-2-5)] (b) {};
\draw[latex-latex] ([yshift=3pt]a.north west) -- node[above] {$\ss k_1$} ([yshift=3pt]a.north east);
\draw[latex-latex] ([yshift=3pt]b.north west) -- node[above] {$\ss k_2$} ([yshift=3pt]b.north east);
\end{tikzpicture}
\\
Q&=0
\label{eq:repQ}
\end{align}
\item All coordinates $u_i$ are zero, and all Pl\"ucker coordinates $x_s$ are zero except for 
$s=\{(i,j): 0\le i< k_j,\ j=1,\ldots,n\}$.
\end{itemize}
\end{prop}
\begin{proof}
We work inside $\M$ first, and consider $T=T_0\times \CC^\times$.
We investigate separately the invariance under the Cartan torus $T_0$ and under $\CC^\times$.

We first consider the latter. A fixed point under this $\CC^\times$ is a $GL(V)$ orbit of $(Q,\Qt,\Phi)$
inside $M^s$ such that for all $t\in\CC^\times$, there exists $A\in GL(V)$ with
\begin{align}
Q&=t^{-\ell} A Q\label{eq:Qfp}
\\
\Qt &= \Qt A^{-1}\label{eq:Qtfp}
\\
\Phi &= t\, A \Phi A^{-1}\label{eq:Phifp}
\end{align}
Denote $V_i=\Im((\Qt \Phi^i)^*)\subset V^*$, $i=1,\ldots,\min(k,\ell)-1$, i.e., the span
of the row vectors $\Qt_j\Phi^i$ of $\Qt\Phi^i$, $j=1,\ldots,n$ (recall that a basis of $W\cong\CC^n$ is fixed).
The stability condition \eqref{eq:stab} is equivalent to the fact that $\sum_{i=0}^{\min(k,\ell)-1} V_i = V^*$.
Deriving from \eqref{eq:Qtfp}--\eqref{eq:Phifp}
\[
\Qt \Phi^iA=t^i \Qt \Phi^i
\]
we conclude that the $V_i$ are the left eigenspaces of $A$, with eigenvalue $t^i$ (or zero). In particular they
are in direct sum: $\bigoplus_{i=0}^{\min(k,\ell)-1} V_i = V^*$. Furthermore, \eqref{eq:Qfp} implies that
$Q=0$ because $A$ has no eigenvalue $t^\ell$.

Next we consider the invariance under the Cartan torus. This implies that the subspaces $W_i=\Im(\Qt\Phi^i)\subset
W\cong \CC^n$ are coordinate subspaces. We can now pick a basis (equivalently, use the remaining gauge
freedom) inside each $V_i$ such that the row vectors $\Qt_j\Phi^i$ are either $0$ or basis vectors
$(0,\ldots,0,1,0,\ldots,0)$. Ordering them as $(\Qt_1,\Qt_1 \Phi,\ldots,\Qt_2,\Qt_2\Phi,\ldots,\Qt_n,\Qt_n\Phi,\ldots)$,
where we only keep the nonzero vectors in this list, 
and recalling that these vectors must span $V^*$ and therefore form a basis,
we find that $\Phi$ is in its Jordan normal form in the dual basis
and are led to the following structure for $\Qt$.
Each $\Qt_i$ is either zero or the topmost basis vector in a Jordan block of $\Phi$ (so that it is
cyclic in that block). Denoting $k_i$ the size of that block or $0$ if $\Qt_i=0$,
$i=1,\ldots,n$, 
we have by definition $\sum_{i=1}^n k_i=k$ and $0\le k_i\le\ell$.

Inversely, given such a sequence $(k_1,\ldots,k_n)$, we define the triplet
as in \eqref{eq:repPhi}--\eqref{eq:repQ}.
One can then reverse the reasoning and show that
it satisfies the invariance properties \eqref{eq:Qfp}--\eqref{eq:Phifp}
and is in the stable set.

Concerning the second description, note that $Q=0$ implies $u_i=0$ for all $i$.
As to the matrix
$\begin{pmatrix}\Qt\\\Qt\Phi\\\vdots\\\Qt\Phi^{k-1}
\end{pmatrix}$, it has exactly one $1$ per column, so that the only nonzero $k\times k$ minor
is obtained by keeping exactly the rows where these $1$s sit. This gives the statement of the proposition concerning
the $x_s$.

We have thus computed $\M^{T}$.
Finally, we show that $\M^{\hat T}=\M^T=\M_1^{\hat T}=\M_1^T$. The fixed points described above are manifestly invariant under
the whole of $\hat T$. Furthermore, because $Q=0$, we have $\M^T \subset \P\subset \M_1$ by \eqref{eq:PM1}.
\end{proof}
We comment on the importance of this proposition. It implies that after localizing, $H_*^{T}(\M_1(k,n,\ell))$ is a free module of
rank equal to the number of fixed points, i.e., the cardinality of 
$\left\{ (k_1,\ldots,k_n):\ 
0\le k_i\le\ell,\ \sum_{i=1}^n k_i=k\right\}$. Now it is easy to see that this is also the dimension
of the $k^{\rm th}$ weight space (counted from the highest weight) of the tensor product $(\CC^{\ell+1})^{\otimes n}$
(e.g., for dimension-counting purposes, $k_i$ can be identified with the number of times one acts with the lowering operator of $\mathfrak{sl}(2)$
on the $i^{\rm th}$ factor of the tensor product).
This is a necessary condition for the identification of $\bigoplus_k H_*^{T}(\M_1(k,n,\ell))$ with the Yangian representation
$\bigotimes_{i=1}^n \CC^{\ell+1}(z_i)$, where each factor is an evaluation representation of $Y(\mathfrak{sl}(2))$
(and $\varphi$ plays the role of homogeneity parameter).

\subsection{The attracting scheme}\label{sec:attsch}
Consider a ``generic'' cocharacter of the Cartan torus $T_0\cong (\CC^\times)^n$.
By generic we mean that it does not lie in any root hyperplane of $GL(n)$.

We are interested in the attracting set
\[
\Z = \bigsqcup_{p\in \M_1^T} \Z_p,
\qquad
\Z_p=\{m\in\M: \lim_{t\to 0}t\cdot m = p\}
\]
\rem[gray]{the definition $\Z=\left\{m\in \M: \lim_{t\to 0}t\cdot m\text{ exists}\right\}$ would be wrong because some points won't converge
to fixed points, cf at $k=\ell=n=2$,
$\Phi=\begin{pmatrix}0&1\\0&0
\end{pmatrix}$,
$Q=\begin{pmatrix}1&0\\0&1
\end{pmatrix}$,
$\Qt=\begin{pmatrix}1&0\\0&-1
\end{pmatrix}$. simpler version: examine the fiber of the singular fixed point,
there's a neutral direction}%
where $t\in\CC^\times$ acts via the cocharacter. $\Z$ only depends on the Weyl chamber of the cocharacter, and
in practice, we shall always make the following choice:
$t\mapsto (t^{a_1},\ldots, t^{a_n})$
where $a_1>a_2>\cdots>a_n$; all other choices are related by Weyl group action.

In the variables $u_i$ and $x_s$ defined in Sect.~\ref{sec:param},
it is easy to see that
\begin{equation}\label{eq:defZ}
\Z = \left\{ (u_i,x_s)\in \M: u_i \text{ strict upper triangular, }i=1,\ldots,k-1\right\}
\end{equation}
Indeed the action of $T_0$ on the $u_i$ is by conjugation, and so $\lim_{t\to 0} u_i=0$ iff $u_i$ is strict upper triangular.
The statement then follows
from the fact that $\P$ (corresponding to the $x_s$ variables) is projective.

We take \eqref{eq:defZ} to be the definition of $\Z$ as a scheme, and call it the {\em attracting scheme.}

Based on computer experiments for $n\le 3$ and various values of $k$ and $\ell$, we conjecture the following
\begin{conj}\label{conj:lagr}
$\Z$ is a Lagrangian subscheme of $\M_1$.
\end{conj}
This conjecture would follow from general arguments if $\M_1$ were smooth symplectic;
but $\M_1$ is not smooth in general. 
\rem[gray]{would a symplectic resolution help? not obviously}

Furthermore, with the same hypotheses as in the previous paragraph, each $\Z_p$ would be an irreducible open subvariety of $\Z$,
and its closure would be an irreducible component of $\Z$. This does not hold in general;
in fact, the smallest counterexample can be found at $n=k=3$, $\ell=2$, where
the attracting set of the fixed point given by the sequence $(1,1,1)$ has three irreducible components.

\subsection{The stable envelope}\label{sec:stabenv}
We consider the stable envelope construction \cite[Sect.~3]{MO-qg} in the case that $\M^T$ is finite, and reformulate
it in a homological language. The analogue of Nakajima quiver variety here is $\bigsqcup_{k=0}^{\ell n} \M_1(k,n,\ell)$, though
for simplicity we work at fixed $k$ for now.

Equivariant localization allows us to expand any class
in $H_*^T(\M_1)$ into the classes of fixed points, allowing denominators. Denote by $\cf_p(\gamma)$ the coefficient of the class
$\gamma$ at the fixed point $p$.

\begin{conj}\label{conj:stab}
There exist unique classes $(S_p)_{p\in\M_1^T}$ in $H_{\dim\Z}^T(\M_1)$ satisfying for all $p\in \M_1^T$:
\begin{enumerate}[(i)]
\item The support of $S_p$ lies inside $\bigsqcup_{q\le p} \Z_q$ (i.e., there exists a class on the latter which is pushed forward to $S_p$).
\item $\cf_p(S_p)=\cf_p(\overline{\Z_p})$.
\item 
$\cf_q(S_p)|_{\varphi=0}=0$
for $q<p$.
\end{enumerate}
\end{conj}
Note that because of the the support from (i) and the degree assumption,
each $S_p$ must be an integer linear combination of classes of the 
irreducible components of $\Z$ (but not necessarily of the $\overline{\Z_q}$).

\begin{rmk}
Condition (ii) is written differently than in \cite[Sect.~3.3.4]{MO-qg}, because it is not obvious how to extend
the notion of polarization to singular fixed points (in the smooth case, our condition (ii) would correspond to the choice of a ``trivial'' polarization,
which is in general not the most ``sign-saving'', cf \cite[Ex.~3.3.2]{MO-qg}).
\end{rmk}

\section{Calculation of the $R$-matrix}\label{sec:R}
\subsection{Definition}
The Weyl group $\mathcal S_n$ of $GL(n)$ naturally acts on $H_*^T(\bigsqcup_k \M_1(k,n,\ell))
=\bigoplus_k H_*^T(\M_1(k,n,\ell))$. We can in particular
consider its action on the stable basis $(S_p)_{p\in \bigsqcup_k\M_1^T(k,n,\ell)}$, resulting in
$\text{Frac}(H_T(\cdot))$-valued matrices denoted $\check R_w$:
\begin{equation}\label{eq:defR}
w\, S_p = \sum_{q\in \bigsqcup_k \M_1^T(k,n,\ell)} (\check R_w)_{q,p} S_q
\end{equation}
and called {\em $R$-matrices} (see \cite{moscowlectures} for a similar definition at $\ell=1$). Taking the sum over $k$ is necessary to realize the tensor
product structure that was already observed at the end of Sect.~\ref{sec:tor}
and which is one of the key features of stable bases. 
In practice however, each value of $k$ provides a diagonal block
of $\check R_w$, and we shall go back to fixing $k$ in what follows.
Also, we shall work inside $H_*^{\hat T}(\M_1(k,n,\ell))$, and only
impose \eqref{eq:subtorus} at the end to restrict to $T$.

A nice feature of $R$-matrices is that they only act
``locally'' in tensor products (really, this is a property of the stable envelope, cf \cite[Cor.~4.1.2]{MO-qg}).
This means that to compute $R$-matrices, we only need to consider the case $n=2$,
and the unique nontrivial element of the Weyl group. We simply denote the corresponding
$R$-matrix $\check R$ (without index).

In the remainder of this section, we thus set $n=2$.
In the language of Sect.~\ref{sec:poisson}, one has $\M_1=\M'$,
so we can write explicitly
\begin{multline*}
\M_1(k,2,\ell) = \Bigg\{
(Q,\Qt,\Phi): \Phi^\ell=0,\ \sum_{i=0}^{\ell-1}\Phi^i Q\Qt \Phi^{\ell-1-i}=0,
\\[-5mm]
\ \tr \Phi^i=0\ (i>0),\ \tr (\Phi^i Q\Qt)=0\ (i\ge0)
\Bigg\}/\!/GL(V)
\end{multline*}
though we shall not need this fact in what follows.
\rem[gray]{seems true in known examples. proof?}
Also note that the dimension formula \eqref{eq:dim} becomes
\begin{equation}\label{eq:dim2}
\dim \M_1(k,2,\ell) = 2 \begin{cases}
k & 0\le k\le\ell
\\
2\ell-k & \ell\le k\le 2\ell
\end{cases}
\end{equation}
$\M_1$ being empty for $k>2\ell$.

We first briefly describe the example of $\ell=2$, and then consider the case $\ell=k$ (which
as we shall argue, is actually the general case).

In particular, we shall prove (for $n=2$) Conj.~\ref{conj:lagr} and Conj.~\ref{conj:stab}.

\subsection{The example of $\texorpdfstring{\ell}{l}=2$}\label{sec:l2ex}
In this subsection, we fix $\ell=2$, so that
$0\le k\le 4$.

$\M_1$ is a point for $k=0,4$ so 
\[
\check R=(1)
\qquad k=0,4
\]

For $k=1,3$, $\M_1 \cong T^*\PP^1$, as expected from the general facts of the previous section.
The calculation of the $R$-matrix
can be found e.g.\ in \cite[Sect.~4.1.2]{MO-qg}; with our conventions, the result reads
\[
\check R = 
\begin{pmatrix}
\frac{\varepsilon}{\varepsilon-z}
&
\frac{z}{\varepsilon-z}
\\
\frac{z}{\varepsilon-z}
&
\frac{\varepsilon}{\varepsilon-z}
\end{pmatrix}
\qquad k=1,3
\]
where we recall that $\varepsilon=-2\varphi$.

\def\y{x_{\begin{smallmatrix}1\\2\end{smallmatrix}}}
\def\x#1#2{x_{\begin{smallmatrix}#1&#2\end{smallmatrix}}}
The only new case is $k=2$. In the parametrization of Sect.~\ref{sec:param}, the variables
are the entries of $u_0$ and $u_1$, and 
$\y$,
$\x11$,
$\x12=\x21$,
$\x22$.
The ideal of equations of $\M_1(2,2,2)$ is generated by
\begin{align*}
&\x21^{2}{-\x11\x22}
&&{u}_{1,1,1}+{u}_{1,2,2}&&{u}_{0,1,1}+{u}_{0,2,2}\\
&{u}_{1,2,2}\x21{-{u}_{1,1,2}\x22}&&{u}_{1,2,1}\x21+{u}_{1,2,2}\x22&&{u}_{1,2,2}\x11{-{u}_{1,1,2}\x21}\\
&{u}_{1,2,1}\x11+{u}_{1,1,2}\x22&&{u}_{0,2,1}\x11+2\,{u}_{0,2,2}\x21{-{u}_{0,1,2}\x22}&&{u}_{1,2,2}\y{-{u}_{0,2,2}\x21}+{u}_{0,1,2}\x22\\
&{u}_{1,2,1}\y{-{u}_{0,2,1}\x21}{-{u}_{0,2,2}\x22}&&{u}_{1,1,2}\y{-{u}_{0,2,2}\x11}+{u}_{0,1,2}\x21\end{align*}
where the first equation is that of $\P(2,2,2)$, the next two are simply the vanishing of the traces of $u_0$, $u_1$,
and the remaining ones are bilinear equations of the form \eqref{eq:bil}.

This variety is singular at the point $u_0=u_1=0$, $\x11=\x12=\x22=0$, which is nothing but the ``middle'' fixed point
$p_{1,1}$. If we remove its fiber under $\M\to \P$, 
we obtain $\M_1^o \cong O(2)\oplus O(-2)\oplus O(-2)$ over $\PP^1$.

The attracting scheme $\Z$ is given in terms of the variables $u_{0,1,2}$, $u_{1,1,2}$ and the same $x_s$ by the 
vanishing of
\begin{align*}
&{u}_{1,1,2}\x22&&{u}_{0,1,2}\x22&&\x21^{2}{-\x11\x22}&&{u}_{1,1,2}\x21&&{u}_{1,1,2}\y+{u}_{0,1,2}\x21
\end{align*}
The irreducible components are
\begin{align*}
\overline{\Z_0}&=\{\x21=\x22=\y=0\}\\
\overline{\Z_1}&=\{u_{1,1,2}=0,\ \x21=\x22=0\}\\
\overline{\Z_2}&=\{u_{0,1,2}=u_{1,1,2}=0,\ \x21^2-\x11\x22=0\}=\P(2,2,2)
\end{align*}
corresponding to the attracting sets of $p_{2,0},p_{1,1},p_{0,2}$ respectively.
$\overline{\Z_0}\cong \CC^2$, $\overline{\Z_1}\cong \CC\times \PP^1$, but $\overline{\Z_2}$ is singular at the fixed point $p_{1,1}$.

By equivariant localization one can expand the $H_*^{\hat T}$
classes of the $\overline{\Z_i}$ in the fixed point basis; the result is
\[
([\overline{\Z_0}],[\overline{\Z_1}],[\overline{\Z_2}])
=
([p_{2,0}],[p_{1,1}],[p_{0,2}])
\begin{pmatrix}
\frac{1}{\left(\varepsilon+z\right)\left(\varphi+\varepsilon+z\right)}&\frac{{-1}}{\left(\varepsilon+z\right)\left(\varphi+z\right)}&\frac{1}{z\left(\varphi+z\right)}\\
0&\frac{1}{\left(\varepsilon+z\right)\left(\varphi+z\right)}&\frac{2}{\left(\varphi{-z}\right)\left(\varphi+z\right)}\\
0&0&\frac{{-1}}{z\left(\varphi{-z}\right)}\end{pmatrix}
\]
where $z=z_1-z_2$. For example, let us derive
the coefficient of $p_{1,1}$ in $\overline{\Z_2}$ (for a fuller justification,
see Sect.~\ref{sec:res} below). In the patch $\y\ne0$ around $p_{1,1}$, $\Z_2$ is the affine scheme
$\{b^2-ac=0\}$ where $a=\frac{\x11}{\y}$, $b=\frac{\x12}{\y}$, $c=\frac{\x22}{\y}$. The coefficient
is simply the weight of the equation, which is twice that of $b$, divided by the products of weights of all
variables, where the weight of $a$ is $\varphi+z$ and the weight of $c$ is $\varphi-z$.
In particular the $2$ of the numerator can be traced back to the fact that $\Z_2$ is singular
at $p_{1,1}$.

We now want linear combinations of these classes whose 
off-diagonal entries vanish at $\varphi=\varepsilon=0$.
This is a simple linear problem, whose solution is unique:
\[
(S_0,S_1,S_2)=
([\overline{\Z_0}],[\overline{\Z_1}],[\overline{\Z_2}])
\begin{pmatrix}
1&1&1
\\
0&1&2
\\
0&0&1
\end{pmatrix}
\]
resulting in
\[
(S_0,S_1,S_2)=
([p_{2,0}],[p_{1,1}],[p_{0,2}])
  \begin{pmatrix}
\frac{1}{\left(\varepsilon+z\right)\left(\varphi+\varepsilon+z\right)}&\frac{{-\varepsilon}}{\left(\varepsilon+z\right)\left(\varphi+z\right)\left(\varphi+\varepsilon+z\right)}&\frac{\varepsilon\,\left(\varphi+\varepsilon\right)}{z\left(\varepsilon+z\right)\left(\varphi+z\right)\left(\varphi+\varepsilon+z\right)}\\
0&\frac{1}{\left(\varepsilon+z\right)\left(\varphi+z\right)}&\frac{2\,\left(\varphi+\varepsilon\right)}{\left(\varepsilon+z\right)\left(\varphi-z\right)\left(\varphi+z\right)}\\
0&0&\frac{{-1}}{z\left(\varphi-z\right)}\end{pmatrix}
\]

The nontrivial element of the Weyl group $\ZZ_2$ reverses the order of the fixed points, and sends
$z$ to $-z$; from this we easily deduce the $k=2$ $R$-matrix:
\[
\check R=
\begin{pmatrix}
\frac{\varepsilon\,\left(\varphi+\varepsilon\right)}{\left(\varepsilon{-z}\right)\left(\varphi+\varepsilon{-z}\right)}&\frac{z\varepsilon}{\left(\varepsilon{-z}\right)\left(\varphi+\varepsilon{-z}\right)}&\frac{{-z\left(\varphi{-z}\right)}}{\left(\varepsilon{-z}\right)\left(\varphi+\varepsilon{-z}\right)}\\
\frac{2\,z\left(\varphi+\varepsilon\right)}{\left(\varepsilon{-z}\right)\left(\varphi+\varepsilon{-z}\right)}&\frac{\varphi\,\varepsilon+\varphi\,z+\varepsilon^{2}+z^{2}}{\left(\varepsilon{-z}\right)\left(\varphi+\varepsilon{-z}\right)}&\frac{2\,z\left(\varphi+\varepsilon\right)}{\left(\varepsilon{-z}\right)\left(\varphi+\varepsilon{-z}\right)}\\
\frac{{-z\left(\varphi{-z}\right)}}{\left(\varepsilon{-z}\right)\left(\varphi+\varepsilon{-z}\right)}&\frac{z\varepsilon}{\left(\varepsilon{-z}\right)\left(\varphi+\varepsilon{-z}\right)}&\frac{\varepsilon\,\left(\varphi+\varepsilon\right)}{\left(\varepsilon{-z}\right)\left(\varphi+\varepsilon{-z}\right)}
\end{pmatrix}
\qquad k=2
\]

Finally, setting $\varepsilon=-2\varphi$ and then redefining $z\mapsto
\varphi z$ for compactness (effectively setting $\varphi=1$), we obtain the
full $R$-matrix:
\[
\check R=\begin{pmatrix}
1\\
&\frac{2}{z+2}&&\frac{-z}{z+2}\\
&&
\frac{2}{\left(z+1\right)\left(z+2\right)}&&\frac{{-2z}}{\left(z+1\right)\left(z+2\right)}&&\frac{z\left(z{-1}\right)}{\left(z+1\right)\left(z+2\right)}\\
&\frac{-z}{z+2}&&\frac{2}{z+2}\\
&&
\frac{{-2z}}{\left(z+1\right)\left(z+2\right)}&&\frac{z^{2}+z+2}{\left(z+1\right)\left(z+2\right)}&&\frac{{-2z}}{\left(z+1\right)\left(z+2\right)}\\
&&&&&\frac{2}{z+2}&&\frac{-z}{z+2}\\
&&
\frac{z\left(z{-1}\right)}{\left(z+1\right)\left(z+2\right)}&&\frac{{-2z}}{\left(z+1\right)\left(z+2\right)}&&\frac{2}{\left(z+1\right)\left(z+2\right)}\\
&&&&&\frac{-z}{z+2}&&\frac{2}{z+2}\\
&&&&&&&&1
\end{pmatrix}
\]
This coincides with the known expression for the spin $1$ $R$-matrix~\cite{ZamFat, KulSkl} (up to signs related to the normalization of basis vectors).
In particular, it satisfies the Yang--Baxter equation:
\[
(\check R(z_2-z_3)\otimes 1)(1\otimes \check R(z_1-z_3))(\check R(z_1-z_2)\otimes 1)
=
(1\otimes \check R(z_1-z_2))(\check R(z_1-z_3)\otimes 1)(1\otimes \check R(z_2-z_3))
\]
where $z$ should be substituted with the argument of each $\check R$. In principle,
this relation should be obtained geometrically from the $n=3$ varieties thanks to the general
properties of the stable envelope; since our varieties are singular, we cannot directly
apply the results of \cite{MO-qg}.

\begin{rmk}Even though $\P(2,2,2)$ is a toric variety, $\M_1(2,2,2)$ is {\em not} a hypertoric variety.
\end{rmk}

\subsection{The general case: $\texorpdfstring{\ell}{l}=k$}\label{sec:general}
Now consider general $\ell$ and $k$. 
The dimension formula \eqref{eq:dim2}
illustrates two simple facts about the dependence of $\M_1(k,2,\ell)$ 
on $k$, $\ell$:
\begin{itemize}
\item As long as $\ell\ge k$, $\M_1(k,2,\ell)$ is actually independent of $\ell$.
Indeed, the definition of $\M_1$ above implies that $\M_1(k,2,\ell)\subset \M_1(k,2,\ell')$
if $\ell\le\ell'$. But since they're irreducible and of the same dimension as long as $\ell,\ell'\ge k$,
they're equal.
\item $\M_1(k,2,\ell)\cong \M_1(2\ell-k,2,\ell)$.
In fact this is a special case of a more general statement for any $n$,
which is provided in App.~\ref{app:isom}, and which is similar to the
Weyl group invariance of Nakajima quiver varieties \cite{Naka-Weyl}.
\end{itemize}

The result is that if $k<\ell$, the parameter $\ell$ is unnecessary and
can be set equal to $k$, and if $k>\ell$, we can set $k'=2\ell-k$ and
we are reduced to the previous case. We shall therefore set $\ell=k$ in what
follows, without loss of generality.

If $n=2$, $k=\ell$, $T$-fixed points of $\M_1$ can be parametrized as follows.
Recall that $p_{k-j,j}$ is the fixed point given by $k_1=k-j$, $k_2=j$, 
i.e., such that the row vectors $\Qt_i$, $i=1,2$,
are cyclic vectors for Jordan blocks of size $k_i$ of $\Phi$. Then $\M_1^T(k,2,k)=\{p_{k-j,j},\ j=0,\ldots,k\}$.

The rest of this section is basically one long calculation, which we summarize now.

\subsection{Nesting of $\texorpdfstring{\P}{P}(k,2,k)$}\label{sec:nest}
There is an embedding $\P(k-1,2,k-1)\hookrightarrow \P(k,2,k)$ which can be described as follows.
\rem[gray]{Maybe replace with something like: Write $V^*=\CC\oplus V'^*$, and denote $v_0$ a non-zero element in $V'^*/ V^\ast$. The stability condition implies that $V^\ast\simeq (\Qt_1, \Qt_2)\CC[\Phi]$. If we define the action of $\Phi$ on $v_0$ as $v_0 \Phi=\Qt_1$, we may write $V'^*\simeq (v_0, \Qt_2)\CC[\Phi]$. PZJ: not convinced it's simpler} 
Write $V^*=\CC\oplus V'^*$,
where $\dim V'=k-1$. Then given an orbit $(\Qt',\Phi')$ in $\P(k-1,2,k-1)$, we associate  to it
\begin{equation}\label{eq:prenest}
\Qt =  \begin{pmatrix}
1&0
\\
0&\Qt'_2
\end{pmatrix}
\qquad
\Phi = \begin{pmatrix}
0&\Qt'_1
\\
0&\Phi'
\end{pmatrix}
\end{equation}
By comparing
the matrices $\begin{pmatrix}\Qt\\\Qt\Phi\\\vdots\\\Qt\Phi^{k-2}\end{pmatrix}$
and
 $\begin{pmatrix}\Qt'\\\Qt'\Phi'\\\vdots\\\Qt'\Phi'^{k-1}\end{pmatrix}$, 
it is easy to see that the rank of the former is $1$ plus
the rank of the latter, so the stability condition is satisfied.
Because $GL(V')\subset GL(V)$, the result sits in a single $GL(V)$ orbit and therefore the map is well-defined.
It is also obviously injective.

In terms of coordinates $x_s$, we can write:
\begin{equation}\label{eq:nestset}
\P(k-1,2,k-1)\overset{\text{set}}{\cong} \P(k,2,k)\cap \{x_{2\ldots2}=0\}
\end{equation}
Indeed, $x_{2\ldots2}=0$ means that $\Qt_2$ is {\em not}\/ a cyclic vector for $\Phi$ acting on the left, 
so that stability requires that $\Qt_1$ is not in $\Qt_2 \CC[\Phi]$. Decomposing $V^*=\CC \Qt_1
\oplus V'^*$ where $V'^*$ contains $\Qt_2\CC[\Phi]$, we find outselves in the setup of the embedding \eqref{eq:prenest}.

In order to turn this into a scheme-theoretic statement, note that
the first column of the matrix $\begin{pmatrix}\Qt'\\\Qt'\Phi'\\\vdots\\\Qt'\Phi'^{k-1}\end{pmatrix}$ has a single $1$ in the first row, which means that all the coordinates $x_s$ such that $(0,1)\not\in s$
are zero. We conclude that
\[
\P(k-1,2,k-1)\subset \P(k,2,k)\cap \{x_s=0\text{ for all } s\text{ such that }(0,1)\not\in s\}
\]
In fact, this is an equality, and by iterating, we obtain
\begin{prop}\label{prop:nest}
There is a scheme isomorphism
\begin{multline}\label{eq:nest}
\P(j,2,j)\cong \P(k,2,k)\cap \{x_s=0\text{ for all } s\text{ such that }(a,1)\not\in s\text{ for some }a< k-j \},
\\ 0\le j\le k
\end{multline}
which to a coordinate $x_s$ on $\P(j,2,j)$ associates the coordinate $x_{s'}$ on $\P(k,2,k)$ where $s'=\{(0,1),\ldots,(k-j-1,1),
\ (k-j+a,1),\ (a,1)\in s,\ (a,2),\ (a,2)\in s\}$.
\end{prop}
We mention this result without proof for now in order to motivate a definition that follows.
The proof will be given after Prop.~\ref{prop:res1}.

In the alternate notation $x_s=x_{s_0\ldots s_{k-1}}$, the effect of the embedding is to add $k-j$ $1$s to the left,
pushing the existing $1$s $k-j$ steps to the right.

\subsection{Attracting sets}
We note that because $\Z(k,2,k)=\Z(k,2,\ell)$ for all $\ell\ge k$, by letting $\ell$ be sufficiently large
in \eqref{eq:defM2}, we can simplify the definition \eqref{eq:defZ} to
\[
\Z(k,2,k)=\{(Q,\tilde Q,\Phi): \Phi^k=0,\ (\Qt\Phi^iQ)_{1,1}=(\Qt\Phi^iQ)_{2,1}=(\Qt\Phi^iQ)_{2,2}=0
\}/\!/GL(k)
\]
We can now determine the structure of $\Z$; to simplify the notation,
we identify the set of fixed points with $\{0,\ldots,k\}$ and thus
write $\Z_j=\Z_{p_{k-j,j}}$.
\begin{prop}\label{prop:dec}
The decomposition of $\Z(k,2,k)$ into (geometric) irreducible components is
\[
\Z(k,2,k) \overset{\text{set}}{=} \bigcup_{j=0}^k \overline{\Z_{j}(k,2,k)},
\qquad \overline{\Z_{j}(k,2,k)} \cong \P(j,2,j) \times \CC^{k-j}
\]
where $\P(j,2,j)$ is embedded in $\P(k,2,k)$ as in Prop.~\ref{prop:nest},
and $\CC^{k-j}=\text{Spec}(\CC[u_{0,1,2},\ldots,u_{k-j-1,1,2}])$ is the fiber of the map $p: \M(k,2,k)\to \P(k,2,k)$ restricted to $\overline{\Z_j(k,2,k)}$.
\end{prop}
\begin{proof}
The proof is inductive. Among the equations of $\Z(k,2,k)$ we have, starting from \eqref{eq:bil},
setting $i=2$, and choosing $s=\{0,\ldots,k-1\}\times \{2\} \cup \{(a,1)\}$ for some $a\in \{0,\ldots,k-1\}$,
\[
u_{a,1,2}x_{2\ldots2}=0,\qquad a=0,\ldots,k-1
\]

This implies the decomposition
\[
\Z(k,2,k)=\P(k,2,k)\cup \Z'(k,2,k)
\]
where according to \eqref{eq:nestset}, $\Z'(k,2,k)\overset{\text{set}}{=}\Z(k,2,k)\cap p^{-1}\P(k-1,2,k-1)$.

Next, we wish to extend the nesting of Sect.~\ref{sec:nest} to include the variables $u_i$, or equivalently before the quotient by
$GL(V)$, the variables $Q$.

We define a map from $\Z(k-1,2,k-1)\times\CC$ to $\Z'(k,2,k)$ as follows. Given an orbit $(Q,\Qt,\Phi)$ and a scalar $\lambda\in\CC$,
$\Qt'$ and $\Phi'$ are defined as in \eqref{eq:prenest}. As to $Q'$, we expand it as
\[
Q'=\begin{pmatrix}
\kappa&\lambda
\\
Q_1&Q_2
\end{pmatrix}
\]
and then set $\kappa=0$;
and need to show that the linear equations satisfied by $Q'$ in $\Z'$ are equivalent to the ones satisfied by $Q=(Q_1,Q_2)$ in $\Z$,
plus $\kappa=0$.

These equations are the strict upper triangularity of $u'_i=\Qt'\Phi'^iQ'$. At $i=0$, it is equivalent to
$\kappa=0$ and $\Qt_2Q_1=\Qt_2Q_2=0$ (in matrix notations,
i.e., pairing between $V^*$ and $V$ implied), and for $i>0$, to $\Qt_1\Phi^{i-1}Q_1=\Qt_2\Phi^i Q_1=\Qt_2\Phi^i Q_2=0$.
These coincide, as expected, with the strict upper triangularity of $u_i=\Qt\Phi^iQ'$, plus $\kappa=0$.

We conclude that the map is an isomorphism. Note that the inverse map is given in terms
of the variables $(x_s,u_i)$ of Sect.~\ref{sec:param} by
$
(x_s,u_i) 
\mapsto ((x_{s'},u_{i-1}),u_{0,1,2})
$
where $s'$ is as in Prop.~\ref{prop:nest}.

Since $p_{k,0}\in \P(k,2,k)$ but $p_{k,0}\not\in \Z'(k,2,k)$, $\P(k,2,k)$ is a geometric component
of $\Z(k,2,k)$. Furthermore, $p_{k,0}$ corresponds to $\Phi$ having a single Jordan block of size $k$, i.e.,
$p_{k,0}\in \M_1^o$ (in particular, it is a smooth point of $\M_1$). Since taking limits can
only decrease the Jordan form of $\Phi$ in dominance order, $\Z_k(k,2,k)\subset \M_1^o$.
On $\M_1^o$, the action of $T_0$ preserves the symplectic form, 
which implies that its attracting set $\Z_k(k,2,k)$ is Lagrangian, irreducible; 
since $\dim \P(k,2,k)=k=\frac{1}{2}\dim \M_1(k,2,k)$, 
we conclude by irreducibility and dimension that $\overline{\Z_k(k,2,k)}=\P(k,2,k)$.
The result follows by induction.
\end{proof}

In particular, $\Z(k,2,k)$ is equidimensional of the expected dimension $k$, 
i.e., the $\Z_j$, which are isotropic as attracting sets of a Hamiltonian flow, 
are of the correct dimension to be Lagrangian;
we conclude that Conjecture~\ref{conj:lagr} is true at $n=2$.

There is a (partial) order on fixed points given by the transitive closure of
\[
p\in \overline{\Z_q}
\ \Rightarrow\ p\le q
\]
(cf \cite[Sec.~3.2.3]{MO-qg}). Here we find that it coincides with the natural (total) order
on $\{0,\ldots,k\}$.

\subsection{Stable subschemes}
The observations above motivate the following definition.
The {\em stable subscheme}\/ $\S_{j}$ associated to the fixed point $p_{k-j,j}$ is
\begin{align}\label{eq:defS}
\S_{j}(k,2,k) &= \Z(k,2,k)\cap p^{-1} \P(j,2,j),\qquad 0\le j\le k
\\\notag
&=
\Z(k,2,k) \cap \{x_s=0\text{ for all } s\text{ such that }(a,1)\not\in s\text{ for some }a\le k-j \}
\end{align}
where in the first line, $\P(j,2,j)$ is embedded in $\P(k,2,k)$ using \eqref{eq:nest}.
Note that this subscheme is {\em not}\/ reduced; set-theretically, it is just
$\S_{j}\overset{\text{set}}{=}\bigsqcup_{0\le i\le j}\Z_i$. For example, $\S_k$ is simply the attracting scheme
$\Z$ itself.

\subsection{Restriction to fixed points}\label{sec:res}
To each fixed point $p_{k-j,j}$ is naturally associated an affine patch $A_j=\{x_s\ne 0\}$
(where $s$ is given as in Prop.~\ref{prop:fp}),
which contains it (and no other). Explicitly, $s=\{(0,1),\ldots,(k-j-1,1),(0,2),\ldots,(j-1,2)\}$.

We now compute the expansion of the classes of attracting varieties $\Z_j(k,2,k)$
and stable subschemes $\S_j(k,2,k)$ into classes of fixed points in localized equivariant homology.
We do so by restricting them to each patch $A_j$.
As warming up, we first consider the analogous calculation for $\P(k,2,k)$.
Denote $z=z_1-z_2$ (the root of $PGL(2)$).
\begin{prop}\label{prop:res1}
For any $0\le j\le k$,
$\P(k,2,k)\cap A_j$ is a (scheme-theoretic) complete intersection; its variables can be chosen to be
$x_{(k-j,1) \cup s\backslash (p,q)}/x_s$ and $x_{(j,2) \cup s\backslash (p,q)}/x_s$ where $(p,q)$ runs over $s$;
and have $\hat T$-weights
\begin{align*}
\wt{x_{(k-j,1) \cup s\backslash (i,1)}/x_s}&=\varphi(-i+k-j), && i=0,\ldots,k-j-1
\\
\wt{x_{(k-j,1) \cup s\backslash (i,2)}/x_s}&=z+\varphi(-i+k-j),&& i=0,\ldots,j-1
\\
\wt{x_{(j,2) \cup s\backslash (i,1)}/x_s}&=-z+\varphi(-i+j),&& i=0,\ldots,k-j-1
\\
\wt{x_{(j,2) \cup s\backslash (i,2)}/x_s}&=\varphi(-i+j),&& i=0,\ldots,j-1
\end{align*}
whereas its equations have $\hat T$-weights $i\varphi$, $i=1,\ldots,k$.
\end{prop}
\begin{proof}
By homogeneity we can assume $x_s=1$.
The idea is that we go back to variables $(\Qt,\Phi)$ and then fix the $GL(V)$ gauge freedom by picking
as basis vectors of $V^*$ the row vectors whose exterior product $x_s$ is computing, namely,
$\Qt_1,\ldots,\Qt_1\Phi^{k-j-1},\Qt_2,\ldots,\Qt_2\Phi^{j-1}$.

We first take care of the two trivial cases $j=0$, $j=k$ (these correspond to the two smooth
fixed points). For say $j=0$,
$\Phi$ takes its Jordan normal form, whereas $\Qt_1$ is fixed to be the first basis vector, and $\Qt_2$
is free; its entries are precisely the numerators of the 
variables in the third row of the statement of the Proposition.
As to entries in the first row, they are all zero, and we can redundantly introduce them, and check that the
equations of their vanishing have the correct $\hat T$-weight. The case $j=k$ is treated similarly.

Now assume $0<j<k$. Then both $\Qt_1$ and $\Qt_2$ are fixed:
\[
\Qt=\begin{tikzpicture}[
baseline=-\the\dimexpr\fontdimen22\textfont2\relax]
\matrix[matrix of math nodes,left delimiter=(,right delimiter=),ampersand replacement=\&] (m){
1\&0\&\ldots\&0\\
\&\&\&\&1\&0\&\ldots\&0\\
};
\node[rectangle,draw,inner sep=-1pt,fit=(m-1-1) (m-1-4)] (a) {};
\node[rectangle,draw,inner sep=-1pt,fit=(m-2-5) (m-2-8)] (b) {};
\draw[latex-latex] ([yshift=3pt]a.north west) -- node[above] {$\ss k-j$} ([yshift=3pt]a.north east);
\draw[latex-latex] ([yshift=3pt]b.north west) -- node[above] {$\ss j$} ([yshift=3pt]b.north east);
\end{tikzpicture}
\]
As to $\Phi$, all but two rows are fixed:
\[
\Phi=
\begin{tikzpicture}[
baseline=-\the\dimexpr\fontdimen22\textfont2\relax]
\matrix[matrix of math nodes,left delimiter=(,right delimiter=),ampersand replacement=\&,nodes in empty cells] (m){
0\&1\\
 \&\ddots\&\ddots\\[1mm]
\& \&0\&1\\
a_0\&\&\node[left]{\ldots};\&a_{k-j-1}\&b_0\&\&\node[left]{\ldots};\& b_{j-1}\\
\&\& \& \&0\&1\\
\&\& \& \&\&\ddots\&\ddots\\
\&\& \& \&\&\&0\&1\\
c_0\&\&\node[left]{\ldots};\&c_{k-j-1}\&d_0\&\&\node[left]{\ldots};\& d_{j-1}\\
};
\node[rectangle,draw,inner sep=-1pt,fit=(m-1-1) (m-3-4)] (a) {};
\node[rectangle,draw,inner sep=-1pt,fit=(m-5-5) (m-7-8)] (b) {};
\draw[latex-latex] ([xshift=3pt]a.north east) -- node[right] {$\ss k-j-1$} ([xshift=3pt]a.south east);
\draw[latex-latex] ([xshift=-3pt]b.north west) -- node[left] {$\ss j-1$} ([xshift=-3pt]b.south west);
\end{tikzpicture}
\]
One easily checks that the $a_i,b_i,c_i,d_i$ are exactly the numerators of the variables in the statement of the Proposition,
in the same order.

The equations defining $\P(k,2,k)$ as a scheme are simply $\tr \Phi^i=0$, $i=1,\ldots,k$. 
\rem[gray]{we're using implicitly the fact that eqs are invariant already}
Since there are $2k$ variables and $k$ equations, and $\dim \P(k,2,k)=k$ (as was already observed
in the proof of Prop.~\ref{prop:dec}),
$\P(k,2,k)\cap A_j$ is a complete intersection. Moreover, the equation $\tr \Phi^i=0$ has $\hat T$-weight
$i\varphi$.

Finally, one checks by direct computation that the variables have the weights as stated in the Proposition.
\end{proof}
\begin{proof}[Proof of Prop.~\ref{prop:nest}]
Recall that in Sect.~\ref{sec:nest}, we stated without proof Prop.~\ref{prop:nest}.
Since the set-theoretic statement follows from applying \eqref{eq:nestset} inductively,
we only need to show that the r.h.s.\ of \eqref{eq:nest}, namely
\[
X=\P(k,2,k)\cap \{x_s=0\text{ for all } s\text{ such that }(a,1)\not\in s\text{ for some }a< k-j' \}
\]
(where we renamed $j$ to $j'$ for convenience)
is reduced in order to conclude that it is isomorphic to the l.h.s.\ $\P(j',2,j')$. We are now in a position to do so.

In the patch $A_j$, $j\le j'$, and with the same parametrization above, the vanishing
of the $x_s$ for which there exists $a<k-j'$ such that $(a,1)\not\in s$ is equivalent to
\[
a_0=\cdots=a_{k-j'-1}=c_0=\cdots=c_{k-j'-1}=0
\]
Due to these zeroes, the characteristic polynomial of $\Phi$ has $k-j'$ trivial zeroes,
and so it is sufficient to impose the $j'$ equations $\tr \Phi^i=0$, $i=1,\ldots,j'$ in order to ensure that $\Phi$
is nilpotent. We have $2j'$ variables, $j'$ equations, $\dim X=j'$, so $X$ is a local complete intersection.

Now specialize to $j=0$. In the patch $A_0$,
it is obvious that the fixed point $p_{k,0}$ is a smooth point of $X$.
Therefore $X$ is generically reduced. Because it is also a local complete intersection, it is reduced.
\end{proof}

\begin{prop}\label{prop:res2}
For any $0\le j\le j'\le k$,
$\overline{\Z_{j'}}\cap A_j$ is a (scheme-theoretic) complete intersection; 
its variables and their $\hat T$-weights are:
\begin{align*}
\wt{x_{(k-j,1) \cup s\backslash (i,1)}/x_s}&=(-i+k-j)\varphi, && i=k-j',\ldots,k-j-1
\\
\wt{x_{(k-j,1) \cup s\backslash (i,2)}/x_s}&=(-i+k-j)\varphi+z,&& i=0,\ldots,j-1
\\
\wt{x_{(j,2) \cup s\backslash (i,1)}/x_s}&=(-i+j)\varphi-z,&& i=k-j',\ldots,k-j-1
\\
\wt{x_{(j,2) \cup s\backslash (i,2)}/x_s}&=(-i+j)\varphi,&& i=0,\ldots,j-1
\\
\wt{u_{i,1,2}}&=i\varphi+\varepsilon+z,&&i=0,\ldots,k-j'-1
\end{align*}
whereas its equations have $\hat T$-weights $i\varphi$, $i=1,\ldots,j'$.

Consequently, one has the following expansion in localized $H_*^{\hat T}(\M_1)$:
\begin{equation}\label{eq:res2}
[\overline{\Z_{j'}}]_{\hat T}=\sum_{j=0}^{j'} \frac{{j'\choose j}}{
\prod_{i=0}^{k-j'-1} (i\varphi+\varepsilon+z)
\prod_{i=k-2j+1}^{k-j}(i\varphi+z)
\prod_{i=2j-k+1}^{j'-k+j}(i\varphi-z)
}
 [p_{k-j,j}]_{\hat T}
\end{equation}
\end{prop}
\begin{proof}
The first statement is a combination of Prop.~\ref{prop:nest}, \ref{prop:dec} and \ref{prop:res1}.

In general, in order to compute the coefficients of the expansion into fixed points in localized
equivariant homology, 
one should take the Hilbert series of the coordinate ring
in each patch around each fixed point (which would be the coefficient in localized equivariant $K$-theory)
and then expand it at first order in the
weights. In the case of local complete intersections, the result is simply that
the coefficient is the ratio of the product of weights of the equations by the product of weights of the variables.
This is the content of the second part of the proposition, after some minor rewriting of the formula.
\end{proof}

\begin{prop}\label{ref:res3}
For any $0\le j\le j'\le k$,
$\S_{j'}\cap A_j$ is a (scheme-theoretic) complete intersection; its variables and their $\hat T$-weights are:
\begin{align*}
\wt{x_{(k-j,1) \cup s\backslash (i,1)}/x_s}&=(-i+k-j)\varphi, && i=k-j',\ldots,k-j-1
\\
\wt{x_{(k-j,1) \cup s\backslash (i,2)}/x_s}&=(-i+k-j)\varphi+z,&& i=0,\ldots,j-1
\\
\wt{x_{(j,2) \cup s\backslash (i,1)}/x_s}&=(-i+j)\varphi-z,&& i=k-j',\ldots,k-j-1
\\
\wt{x_{(j,2) \cup s\backslash (i,2)}/x_s}&=(-i+j)\varphi,&& i=0,\ldots,j-1
\\
\wt{u_{i,1,2}}&=i\varphi+\varepsilon+z,&&i=0,\ldots,k-j-1
\end{align*}
whereas its equations have $\hat T$-weights $i\varphi$, $i=1,\ldots,j'$,
and $i\varphi+\varepsilon$, $i=j,\ldots,j'-1$.

Consequently, one has the following expansion in localized $H_*^{\hat T}(\M_1)$:
\begin{equation}\label{eq:res3}
[\S_{j'}]_{\hat T}=\sum_{j=0}^{j'} \frac{{j'\choose j} \prod_{i=j}^{j'-1}(i\varphi+\varepsilon)}{
\prod_{i=0}^{k-j-1} (i\varphi+\varepsilon+z)
\prod_{i=k-2j+1}^{k-j}(i\varphi+z)
\prod_{i=2j-k+1}^{j'-k+j}(i\varphi-z)
} [p_{k-j,j}]_{\hat T}
\end{equation}
\end{prop}
\begin{proof}
The proof is very similar to that of Prop.~\ref{prop:res1} and \ref{prop:res2}, and we sketch it only. 
We use the same gauge fixing as in the proof of the former, and then impose \eqref{eq:defS}; we obtain,
as in the proof of Prop.~\ref{prop:nest}:
\[
a_0=\cdots=a_{k-j'-1}=c_0=\cdots=c_{k-j'-1}=0
\]
and so we only need to impose the $j'$ equations $\tr \Phi^i=0$, $i=1,\ldots,j'$.

Now we introduce the variables $Q$.
Imposing that the $u_i=\Qt\Phi^i Q$ are strict upper triangular results in the following structure:
\[
Q=\begin{pmatrix}
0&u_{0,1,2}
\\
\vdots&\vdots
\\
0&u_{k-j-1,1,2}
\\
0&0
\\
\vdots&\vdots
\\
0&0
\end{pmatrix}
\]
as well as the equations
\[
\sum_{r=j}^{j'-1} (-1)^r x_{(i,2)\cup s\backslash (r,1)} u_{r,1,2}=0,\qquad i=j,\ldots,j'-1
\]
(which we recognize as a special case of \eqref{eq:bil}).
Now we have in total $2j'+k-j$ variables and $2j'-j$ equations, which matches with the dimension $k$ of $\S_{j'}$,
hence the complete intersection property. The second part of the Proposition is justified in the same
way as that of Prop.~\ref{prop:res2}.
\end{proof}

Since $\S_j\overset{\text{set}}{=}\bigcup_{i=0}^j\overline{\Z_i}$, we know that its class in $H_*^{\hat T}(\M_1)$ is a linear combination
of the classes of the $\overline{\Z_i}$ (with positive integer coefficients, the so-called multiplicities). 
Thanks to the formulae above, we can be more specific:
\begin{cor}
One has the linear relations:
\begin{equation}\label{eq:linrel}
[\S_j]_{\hat T}=\sum_{i=0}^j {j\choose i} [\overline{\Z_i}]_{\hat T}
\end{equation}
\end{cor}
The proof is a direct calculation (checking that l.h.s.\ and r.h.s.\ have
same residues as rational functions of $z$), and we shall omit it.

\begin{prop}\label{prop:stabcl}
Conj.~\ref{conj:stab} is true at $n=2$, with the choice of classes
$S_j = [\S_j]_T$.
\end{prop}
\begin{proof}
First, the uniqueness at $n=2$ follows from the fact that all attracting sets are irreducible
(and from the already noted fact that the $S_p$ must be linear combinations of classes of irreducible
components of $\Z$).
Indeed, the matrix of restrictions of (closures of) attracting sets at fixed points is a (square)
upper triangular matrix (explicitly it is $z^{-k}$ times a binomial matrix). Its entries on the diagonal
are nonzero, and therefore the solution to property (iii) of Conj.~\ref{conj:stab} is uniquely
determined by inverting that matrix.

Next, we prove that the $S_j = [\S_j]_T$ do satisfy all the properties of Conj.~\ref{conj:stab}.
Property (i) follows from the definition of the $\S_j$ (or from \eqref{eq:linrel}).
Property (ii) follows from the fact that the multiplicity of $\overline{\Z_j}$ in $\S_j$ is $1$
(\eqref{eq:linrel} again).
Finally, recalling that $\varepsilon=-\ell\varphi$,
when we set $\varepsilon=\varphi=0$ in \eqref{eq:res3},
the coefficient of the fixed point $j$ in $[\S_{j'}]_T$ vanishes as long as $j<j'$.
\end{proof}

\subsection{$R$-matrix}\label{sec:RR}
Finally, we derive the $R$-matrix from the calculations of the previous section.
Define the upper triangular matrix based on \eqref{eq:res3}:
\rem[gray]{here I follow [MO]'s convention, NOT my standard convention as in the m2 files
which is to transpose the $R$-matrix}
\begin{align}\label{S1}
S_{jj'}&=\cf_{p_{k-j,j}}([\S_{j'}]_{\hat T})
\\ \nonumber
&=\frac{{j'\choose j} \prod_{r=j}^{j'-1}(r\varphi+\varepsilon)}{
\prod_{r=0}^{k-j-1} (r\varphi+\varepsilon+z)
\prod_{r=k-2j+1}^{k-j}(r\varphi+z)
\prod_{r=2j-k+1}^{j'-k+j}(r\varphi-z)
}
\end{align}

Similarly, consider
\begin{align*}
\tilde S_{jj'}&=\cf_{p_{k-j,j}}([w \S_{j'}]_{\hat T})
\\
&=S_{k-j\,j'}|_{z\to -z}
\end{align*}
where in the second line we used the fact that $w$ (the nontrivial element of the Weyl group $\ZZ_2$) 
sends the $j^{\rm th}$ fixed point to the $(k-j)^{\rm th}$ one,
and sends $z$ to $-z$.

From the definition \eqref{eq:defR},
\[
\check R = S^{-1} \tilde S
\]
in which \eqref{eq:subtorus} must be imposed.
(More traditionally, one uses the $R$-matrix
$R= P \check R$ where here $P_{ij}=\delta_{i,k-j}$,
and then the expression above provides a lower/upper decomposition for $R$, namely
$R=LU$ with $L=PS^{-1}P$, $U=S|_{z\to -z}$.) 

Explicitly, the matrix $S^{-1}$, expressing
the expansion of fixed points into the stable basis, is given by 
\begin{equation}\label{Sinv1}
S^{-1}_{ij}=
    {j\choose i} \prod_{r=i}^{j-1}(r\varphi+\varepsilon)\prod_{r=0}^{k-j-1}(r\varphi +\varepsilon+z)
    \prod_{r=k+1-j-i}^{k-j}(r\varphi+z)
\end{equation}
\rem[gray]{(this formula can be proven by computing explicitly $S^{-1}S$
as a rational function of $z$ whose residues all vanish),}
(the proof is given in Appendix~\ref{app:Sinv}), so that together,
\begin{equation}\label{eq:finalR}
\check R_{ij'}=\!\!\!\!\!\!\sum_{j=\max(i,k-j')}^{k}\!\!\!\!\!\!
 \frac{   {j\choose i} \prod_{r=i}^{j-1}(r\varphi+\varepsilon)\prod_{r=0}^{k-j-1}(r\varphi +\varepsilon+z)
    \prod_{r=k+1-j-i}^{k-j}(r\varphi+z)
{j'\choose k-j} \prod_{r=k-j}^{j'-1}(r\varphi+\varepsilon)}{
\prod_{r=0}^{j-1} (r\varphi+\varepsilon-z)
\prod_{r=2j-k+1}^{j}(r\varphi-z)
\prod_{r=k-2j+1}^{j'-j}(r\varphi+z)
}
\end{equation}
In order to obtain the spin $\ell/2$ $R$-matrix, one should substitute $\varepsilon=-\ell\varphi$;
note however that the expression is valid as is for the $R$-matrix in a Verma module (``bosonic'') representation
of $Y(\mathfrak{sl}(2))$, when $\varepsilon/\varphi$ is generic (i.e., not integer).

It is interesting to compare our expression \eqref{eq:finalR} to existing ones. The general formula for the rational $\mathfrak{sl}_2$ $R$-matrix in the case of arbitrary spin, written as a function of the Casimir invariant, was obtained as early as in~\cite{KRS} (it is nicely reviewed in the lectures~\cite{FaddeevLect}). There is also a representation for the $R$-matrix in coherent state basis~\cite{Sklyanin}, where the $R$-matrix acts in a tensor product of two finite-dimensional representations, and another one, where the action is via a matrix-valued differential operator in the tensor product of a finite-dimensional representation and a generic Verma module~\cite{ChDer}. However these expressions are not explicit enough to be compared
with ours. The general formula
for the rank 1 {\em trigonometric}\/ $R$-matrix in arbitrary spin, written in the natural spin-up/spin-down basis, was first obtained
by Mangazeev in \cite{Manga-R}. One can
easily take the rational limit (corresponding to going from $U_q(\widehat{\mathfrak{sl}(2)})$ to
$Y(\mathfrak{sl}(2))$), resulting in an expression which is very similar, though not trivially equivalent,
to \eqref{eq:finalR}. More precisely, one can conjecture the following: first reindex
\[
R_{i,j}^{i',j'}=\check R_{i,j'}\qquad i+j=i'+j'=k
\]
Then, setting
$I=J=-\varepsilon$ in \cite[Eq.~(1.1)]{Manga-R}, as well as $\lambda=z$, and then taking the limit $q\to1$,
one has up to overall normalization
\[
R^{\text{Mangazeev}}{}_{i,j}^{i',j'}
=
(-1)^{i+i'}
\frac{\Gamma(\varepsilon/\varphi+i)\Gamma(\varepsilon/\varphi+j)}{\Gamma(\varepsilon/\varphi+i')\Gamma(\varepsilon/\varphi+j')}
R_{i,j}^{i',j'}
\]
Note that the prefactor is simply a diagonal conjugation, which can be blamed on a different normalization of the
basis vectors. This relation can presumably be proven using appropriate hypergeometric identities.

\appendix
\section{$\ZZ_2$ invariance of $\M$ and $\M_1$}\label{app:isom}
\noindent
In this section we denote the dual weight by $k^\vee=n\ell-k$.
\begin{lem}
There are two isomorphisms: \\$\M(k, n, \ell)\simeq \M(k^\vee, n, \ell)$ and $\M_1(k, n, \ell)\simeq \M_1(k^\vee, n, \ell)$.
\end{lem}
\begin{proof}
Remembering that $\Qt: \CC^k\to \CC^n$, we form the matrix $\At: \CC^k\to\CC^{\ell n}$ as follows:
\bea\label{Astruct}
\At=\begin{pmatrix}\Qt\\\Qt\Phi\\\vdots\\\Qt\Phi^{\ell-1}
\end{pmatrix}
\eea
(We recall that $\Phi^\ell=0$.) Earlier we introduced the stability condition $\mathrm{rank}\,\At=k$. The dimension of the kernel of $\At$ is therefore $k^\vee$. We define a matrix $\Bt:  \CC^{k^\vee}\to \CC^{\ell n}$ by the equation
\bea\label{ABduality}
\Bt^T\At=0
\eea
and the non-degeneracy requirement $\mathrm{rank}\,\Bt=k^\vee$. The matrix $\At$ was defined up to right multiplication by an element of $GL(k, \CC)$, and the matrix $\Bt$ is accordingly defined up to right multiplication by the element of $GL(k^\vee, \CC)$.

Now, due to the structure~(\ref{Astruct}) of the matrix $\At$, multiplication by $\Phi$ from the right essentially amounts to shifting the entries of $A$ upwards by $n$ rows, and filling the lower rows with zeros. Denoting such a shift matrix by $S_n$, we write
\bea\label{shiftm}
 \At \Phi=S_n \At \,.
\eea
Now, multiplying (\ref{ABduality}) by $\Phi$ from the right, we obtain
\bea
 \Bt^T\At\Phi=\Bt^T S_n \At=0\,
\eea
hence $ S_n^T\Bt\in \mathrm{Ker}(\At^T)$ (note that $S_n^T$ acts on $\Bt$ by shifting by $n$ rows downwards). Since $\mathrm{Ker}(\At^T)$ is spanned by the columns of $\Bt$, we may write
\bea\label{SnB}
 S_n^T\Bt=\Bt \Phi^{\vee} ,\quad\quad\quad \Phi^{\vee}\in\mathrm{End}(\CC^{k^\vee})\,.
\eea
Since $(S_n)^\ell=0$, we find that $(\Phi^{\vee})^\ell=0$. Moreover, we can rephrase~(\ref{SnB}) as the requirement that $\Bt$ has the form
\bea\label{Bstruct}
\Bt=\begin{pmatrix}\Qt^\vee (\Phi^\vee)^{\ell-1} \\\Qt^\vee (\Phi^\vee)^{\ell-2}\\\vdots\\\Qt^\vee
\end{pmatrix}
\eea
Therefore we have defined a map
\bea\label{Z2map1}
\{\Qt, \Phi\}^s/GL(k, \CC) \leftrightarrow \{\Qt^\vee, \Phi^\vee\}^s/GL(k^\vee, \CC)
\eea
We have therefore proved the duality of projective retracts $\mathfrak{P}(k, n, \ell)\simeq \mathfrak{P}(k^\vee, n, \ell)$.

We now wish to incorporate the variable $Q$ in the above discussion. To this end, we form the $\Phi$-tower of $Q$ that we will call $A: \CC^{\ell n}\to\CC^k$, analogously to the way it was done in~(\ref{Astruct}) for the $\Qt$-variable:
\bea\label{ATstruct}
A=\{\Phi^{\ell-1} Q , \ldots,  \Phi Q ,   Q\}\,.
\eea
In this case the NS-equation $\sum\limits_{i+j=\ell-1}\,\Phi^{i}\,Q\,\widetilde{Q}\,\Phi^{j}=0$ may be put in the form
\bea\label{NSmomap}
A\widetilde{A}=0\,,
\eea
which is the NS-generalization of the equation $Q\widetilde{Q}=0$ in the hyper-K\"ahler case. To complete the duality, we impose the following condition:
\bea\label{AtBtduality}
(\At A)^T=\Bt B\,.
\eea
Note that in both sides we have matrices of size $\ell n\times \ell n$. A simple consistency check of the above equations is obtained by multiplying from the left by $\At^T$. The l.h.s. is then $(\At A \At)^T=0$ by (\ref{NSmomap}), and the r.h.s. is $\At^T\Bt B=0$ by (\ref{ABduality}).

Since $\Bt$ has already been defined as (\ref{Bstruct}) from (\ref{ABduality}), one can view the equation~(\ref{AtBtduality}) as an equation for $B$. Multiplying (\ref{AtBtduality}) by $\Bt$ from the right and using~(\ref{ABduality}), we get $\Bt B \Bt=0$. Since $B\Bt\in \mathrm{End}(\CC^{k^\vee})$ and $\mathrm{rank}(\Bt)=k^\vee$, it follows that
\bea\label{NSdual}
B \Bt=0\,,
\eea
which is an analogue of (\ref{NSmomap}).

It remains to show that $B$, just like $A$, may be written as a $\Phi$-tower of a basic element $Q^\vee$. To this end we multiply $(\At A)^T$ by $S_n^T$ from the right and use~(\ref{shiftm}): $(\At A)^T S_n^T=(S_n\At A)^T=(\At \Phi A)^T=(\At A S_n)^T=S_n^T(\At A )^T$. Here we have used $\Phi A=A S_n$. By (\ref{AtBtduality}) this implies $\Bt B S_n^T=S_n^T \Bt B$. Using~(\ref{SnB}), $\Bt (B S_n^T-\Phi^{\vee} B)=0$. Since $\Bt: \CC^{k^\vee}\to \CC^{\ell n}$ and $\mathrm{rank}(\Bt)=k^\vee$, we get $B S_n^T-\Phi^{\vee} B=0$, hence
\bea\label{BTstruct}
B=\{ Q^\vee , \ldots,  (\Phi^\vee)^{\ell-2} Q^\vee ,  (\Phi^\vee)^{\ell-1} Q^\vee\}\,.
\eea
Thus $\M(k, n, \ell)\simeq \M(k^\vee, n, \ell)$.

In order to prove the isomorphism $\M_1(k, n, \ell)\simeq \M_1(k^\vee, n, \ell)$, we will show that, under the map~(\ref{Z2map1}), we have $\M^o_1(k, n, \ell)\leftrightarrow \M^o_1(k^\vee, n, \ell)$.

First we set, as earlier, $k=q\ell+r$, where $0\leq r <\ell$. If $\Phi \in \mathcal{O}_1(k, \ell)$, then $\Phi$ has the Jordan form $(\underbracket[0.6pt][0.6ex]{\ell, \ldots, \ell,}_{q\,\textrm{times}} r)$. This may be reformulated as follows. Consider the truncated matrix $\At_j: \CC^k\to \CC^{(\ell-j+1)n}$, defined as
\bea
\At_j=\begin{pmatrix}\Qt\Phi^{j-1}\\\vdots\\\Qt\Phi^{\ell-1}
\end{pmatrix}
\eea
Since $\Qt$ contains the generators of all Jordan towers, $\mathrm{Span}(\textrm{rows of }\At_j)\simeq  \mathrm{Im}(\Phi^{j-1})$. From the Jordan form of $\Phi$ it follows that
\begin{eqnarray}
&&\mathrm{rank}\,\At_j=\mathrm{dim}\;\mathrm{Im}(\Phi^{j-1})=(\ell-j+1)q \quad \textrm{for}\quad j>r\,, \\
&&\mathrm{rank}\,\At_j=\mathrm{dim}\;\mathrm{Im}(\Phi^{j-1})=(\ell-j+1)q+(r-j+1) \quad  \textrm{for}\quad j\leq r\,.
\end{eqnarray}
For the dimension of the null-space we therefore get $\mathrm{dim}\;\mathrm{Ker}\,\At_j^T=(\ell-j+1)n-\mathrm{rank}\,\At_j$. From the fact that the matrices $\At_j$ are nested inside each other it follows that $\mathrm{Ker}\,\At_\ell^T \subset \mathrm{Ker}\,\At_{\ell-1}^T\subset \cdots \subset \mathrm{Ker}\,\At^T$.

Next we consider the truncated $\Bt$-matrices $\Bt_j: \CC^{k^\vee}\to \CC^{n j}$
\bea
\Bt_j=\begin{pmatrix}\Qt^\vee (\Phi^\vee)^{j-1} \\ \vdots\\\Qt^\vee
\end{pmatrix}
\eea
The vectors spanning $\mathrm{Ker}\,\At_j^T$ enter in some of the columns of $\Bt_{\ell-j+1}$, and accordingly the vectors spanning $\mathrm{Ker}\,\At_{j-1}^T$ enter in some of the columns of the matrix $\Bt_{\ell-(j-1)+1}$, which is obtained from $\Bt_{\ell-j+1}$ by adding the additional rows $\Qt^\vee (\Phi^\vee)^{\ell-j+1}$ on top. Moreover we have the embedding $\mathrm{Ker}\,\At_j^T\subset \mathrm{Ker}\,\At_{j-1}^T$. The vectors spanning $\mathrm{Ker}\,\At_{j-1}^T \setminus \mathrm{Ker}\,\At_j^T$ remain linearly independent even upon truncation to the `upper block' $\Qt^\vee (\Phi^\vee)^{\ell-j+1}$. 
Therefore
\bea\label{rankeq1}
\mathrm{rank}(\Qt^\vee (\Phi^\vee)^{\ell-j+1})\geq \mathrm{dim}(\mathrm{Ker}\,\At_{j-1}^T \setminus \mathrm{Ker}\,\At_j^T)\,.
\eea
Using
\bea
\mathrm{dim}(\mathrm{Ker}\,\At_{1}^T)=\ell(n-q)-r,\quad\quad \mathrm{dim}(\mathrm{Ker}\,\At_{2}^T)=(\ell-1)(n-q)-\mathrm{max}(r-1,0)\,,
\eea
we find from~(\ref{rankeq1}) for $j=2$:
\bea\label{rankeq2}
\mathrm{rank}(\Qt^\vee (\Phi^\vee)^{\ell-1})\geq n-q+\mathrm{max}(r-1,0)-r=\begin{cases} 
   n-q\quad\textrm{if}\quad r=0 \\
   n-q-1   \quad\textrm{if}\quad r>0
  \end{cases}
\eea
Since $k^\vee=n\ell-k=(n-q-1)\ell+\ell-r$, there can be at most $n-q-1$ or $n-q$ Jordan towers of height $\ell$ for $r>0$ and $r=0$, respectively. As the number of such maximal towers is measured by $\mathrm{rank}(\Qt^\vee (\Phi^\vee)^{\ell-1})$, the inequality~(\ref{rankeq2}) shows that this bound is saturated. This implies that $\Phi^\vee$ has the Jordan structure $(\underbracket[0.6pt][0.6ex]{\ell, \ldots, \ell,}_{n-q-1\,\textrm{times}} \ell-r)$, i.e. $\Phi^\vee\in \mathcal{O}_1(k^\vee, \ell)$.

Now we need to show that $Q^\vee \Qt^\vee \in T_{\Phi^\vee} \mathcal{O}_1(k^\vee, \ell)$. Since the defining equations of $\mathcal{O}_1(k^\vee, \ell)$ are $(\Phi^\vee)^\ell=0$ and $\tr((\Phi^\vee)^i)=0$ for all $i$, we need to show that $
{d\over d\epsilon}(\Phi+\epsilon Q^\vee \Qt^\vee)^\ell\big|_{\epsilon=0}=0$ and $\tr((\Phi^\vee)^{i-1}Q^\vee \Qt^\vee)=0$ for all $i$. The first part follows from~(\ref{NSdual}), if one recalls the definitions~(\ref{Bstruct}) and~(\ref{BTstruct}). To prove the second part, we look more closely at the relation~(\ref{AtBtduality}). The $\ell n\times \ell n$-matrix $\At A$ may be viewed as an $\ell\times \ell$-matrix with the entries $(\At A)^T_{ij}=\Qt \Phi^{\ell-1+j-i} Q$, analogously $(\Bt B)_{ij}=\Qt^\vee (\Phi^\vee)^{\ell-1+j-i} Q^\vee$. Therefore~(\ref{AtBtduality}) is equivalent to $\Qt \Phi^{\ell-1+j-i} Q=\Qt^\vee (\Phi^\vee)^{\ell-1+j-i} Q^\vee$ for all $i, j$. Since $(Q, \Qt, \Phi)\in \M^o_1(k, n, \ell)$, $\tr(Q\Qt \Phi^{i})=0$ is satisfied for all $i$. It then follows that $\tr(Q^\vee\Qt^\vee (\Phi^\vee)^{i})=0$ for all $i$, as required. Therefore $Q^\vee \Qt^\vee \in T_\Phi \mathcal{O}_1(k^\vee, \ell)$ and $(Q^\vee, \Qt^\vee, \Phi^\vee)\in \M^o_1(k^\vee, n, \ell)$.


\end{proof}

\section{The inverse matrix $S^{-1}$}\label{app:Sinv}

To prove the formula~(\ref{Sinv1}) for $S^{-1}$, we observe that $S^{-1}S$ is a rational function of $z$. We will now compute the residues of $(S^{-1}S)_{ij'}$ at the potential poles and show that they vanish.

First, changing $r\to-r$ in the third factor in the denominator of~(\ref{S1}), we simplify the expression for $S_{jj'}$:
\begin{align*}
S_{jj'}=
\frac{(-1)^{j'-j}{j'\choose j} \prod_{r=j}^{j'-1}(r\varphi+\varepsilon)}{
\prod_{r=0}^{k-j-1} (r\varphi+\varepsilon+z)
\prod_{r=k-j-j',\;r\neq k-2j}^{k-j}(r\varphi+z)
}
\end{align*}
Using this, as well as (\ref{Sinv1}), we find
\bea\label{SinvS1}
\sum\limits_{j=i}^{j=j'}\,S^{-1}_{ij} S_{jj'}=\prod\limits_{r=i}^{j'-1}(r \varphi+\epsilon)\,\sum\limits_{j\in\mathbb{Z}}\,{j' \choose j} {j\choose i}\frac{(-1)^{j'-j}((k-2j)\varphi+z)}{\prod_{r=k-j-j'}^{k-j-i} (r\varphi+z)}
\eea
We have extended the sum over $j$ to the integers, as the binomial coefficients suppress all terms outside of $i\leq j\leq j'$. Clearly, the potential poles of the r.h.s. are at $z=-n \varphi, n\in\mathbb{Z}$. Note that there are exactly $j'-i+1$ factors in the denominator of each term: setting $r=k-j-j'+r'$, these factors correspond to $r'=0\ldots j'-i$. The $\tilde{r}$-th factor vanishes at $z=-n \varphi$ when $k-j-j'+\tilde{r}=n$, i.e. this happens in the $j$-th term in the sum, where $j=k-j'+\tilde{r}-n$. The residue at this point is the sum of $j'-i+1$ terms:
\begin{align*}
\underset{z=-n \varphi}{\mathrm{res}}\left(S^{-1}S\right)_{ij'}=\prod\limits_{r=i}^{j'-1}(r \varphi+\epsilon)\,\sum\limits_{\tilde{r}=0}^{j'-i}{j' \choose j} {j\choose i}\frac{(-1)^{j'-j}(k-2j-n)\varphi}{\prod_{r'=0, r'\neq \tilde{r}}^{j'-i} ((r'-\tilde{r})\varphi)}\big|_{j=k-j'+\tilde{r}-n}=\\
=\prod\limits_{r=i}^{j'-1}(r \varphi+\epsilon)\,\frac{(-1)^{k-n}}{\varphi^{j'-i-1}}\sum\limits_{\tilde{r}\in\mathbb{Z}}{j' \choose j} {j\choose i}\frac{(k-2j-n)}{\tilde{r}! (j'-i-\tilde{r})!}\big|_{j=k-j'+\tilde{r}-n}
\end{align*}
We have again extended the summation to the integers, due to the presence of the factorials in the denominator. The factor $(k-2j-n)\big|_{j=k-j'+\tilde{r}-n}$ in the numerator is skew-symmetric under the change $\tilde{r}\to 2j'-(k-n)-\tilde{r}$. At the same time, one can check that the product of factorials and binomial coefficients in the sum is symmetric under this change. Therefore $\underset{z=-n \varphi}{\mathrm{res}}\left(S^{-1}S\right)_{ij'}=0$.

It is obvious from~(\ref{SinvS1}) that $(S^{-1}S)_{ij'}$ vanishes for $z\to\infty$ if $j'-i>0$ and approaches~$1$ if $j'=i$. We conclude that $(S^{-1}S)_{ij'}=\delta_{ij'}$.

\gdef\MRshorten#1 #2MRend{#1}%
\gdef\MRfirsttwo#1#2{\if#1M%
MR\else MR#1#2\fi}
\def\MRfix#1{\MRshorten\MRfirsttwo#1 MRend}
\renewcommand\MR[1]{\relax\ifhmode\unskip\spacefactor3000 \space\fi
\MRhref{\MRfix{#1}}{{\scriptsize \MRfix{#1}}}}
\renewcommand{\MRhref}[2]{%
\href{http://www.ams.org/mathscinet-getitem?mr=#1}{#2}}
\bibliographystyle{amsplainhyper}
\bibliography{biblio}
\end{document}